\documentclass[a4paper,UKenglish,cleveref, autoref, thm-restate,]{lipics-v2021}
\usepackage{algorithm}
\usepackage{edtable}
\usepackage{algpseudocode}
\usepackage{todonotes}
\usepackage[inline]{enumitem}

\nolinenumbers

\hideLIPIcs  %uncomment to remove references to LIPIcs series (logo, DOI, ...), e.g. when preparing a pre-final version to be uploaded to arXiv or another public repository

\def\RR{{\mathbb R}}
\def\ZZ{{\mathbb Z}}
\def\EE{{\mathbb E}}
\def\NN{{\mathbb N}}

\def\cD{{\mathcal D}}

\def\cX{{\mathcal{X}}}
\def\cR{{\mathcal{R}}}
\def\cF{{\mathcal{F}}}
\def\cG{{\mathcal{G}}}
\def\poly{{\mathrm {poly}}}
\def\Pr{{\mathrm Pr}}

\def \exp{{\mathrm exp}}

\def \ppath{{\mathcal P}}
\def \ptree{{\mathcal T}}
\def \stree{{T}}
\newcommand{\eps}{\varepsilon}

\title{Space-Efficient Approximate Spherical Range Counting in High Dimensions} %TODO Please add

\author{Andreas Kalavas}{Carnegie Mellon University, Pittsburgh, USA}{akalavas@andrew.cmu.edu}{https://orcid.org/0009-0004-8931-108X}{Supported in part by the U.S. National Science Foundation (NSF) Grants ECCS-2145713, CCF-2403194, CCF-2428569, and ECCS-2432545.}%TODO mandatory, please use full name; only 1 author per \author macro; first two parameters are mandatory, other parameters can be empty. Please provide at least the name of the affiliation and the country. The full address is optional. Use additional curly braces to indicate the correct name splitting when the last name consists of multiple name parts.

\author{Ioannis Psarros}{Archimedes, Athena Research Center, Athens, Greece}{ipsarros@athenarc.gr}{https://orcid.org/0000-0002-5079-5003}{This work has been partially supported by project MIS 5154714 of the National Recovery and Resilience Plan Greece 2.0 funded by the European Union under the NextGenerationEU Program.}

\authorrunning{A. Kalavas and I. Psarros} %TODO mandatory. First: Use abbreviated first/middle names. Second (only in severe cases): Use first author plus 'et al.'

\Copyright{Andreas Kalavas and Ioannis Psarros} %TODO mandatory, please use full first names. LIPIcs license is "CC-BY";  http://creativecommons.org/licenses/by/3.0/

 %TODO mandatory: Please choose ACM 2012 classifications from https://dl.acm.org/ccs/ccs_flat.cfm 

 %TODO mandatory; please add comma-separated list of keywords

\category{} %optional, e.g. invited paper

\relatedversion{} %optional, e.g. full version hosted on arXiv, HAL, or other respository/website
\acknowledgements{We would like to thank Dimitris Fotakis for valuable discussions on the topic.}%optional

\EventEditors{John Q. Open and Joan R. Access}
\EventNoEds{2}
\EventLongTitle{42nd Conference on Very Important Topics (CVIT 2016)}
\EventShortTitle{CVIT 2016}
\EventAcronym{CVIT}
\EventYear{2016}
\EventDate{December 24--27, 2016}
\EventLocation{Little Whinging, United Kingdom}
\EventLogo{}
\SeriesVolume{42}
\ArticleNo{23}
\begin{document}

\maketitle

\begin{abstract} 
   We study the following range searching problem in high-dimensional Euclidean spaces: given a finite set $P\subset \mathbb{R}^d$, where each $p\in P$ is assigned a weight $w_p$, and radius $r>0$, we need to preprocess $P$ into a  data structure such that when a new query point $q\in \mathbb{R}^d$ arrives, the data structure reports the cumulative weight of points of $P$ within Euclidean distance $r$ from $q$. Solving the problem exactly seems to require space usage that is exponential to the dimension, a phenomenon known as the curse of dimensionality. Thus, we focus on approximate solutions where points up to $(1+\varepsilon)r$ away from $q$ may be taken into account, where $\varepsilon>0$ is an input parameter known during preprocessing. We build a data structure with near-linear space usage, and query time in $n^{1-\Theta(\varepsilon^4/\log(1/\varepsilon))}+t_q^{\varrho}\cdot n^{1-\varrho}$, for some $\varrho=\Theta(\varepsilon^2)$, where $t_q$ is the number of points of $P$ in the ambiguity zone, i.e., at distance between $r$ and $(1+\varepsilon)r$ from the query $q$. To the best of our knowledge, this is the first data structure with efficient space usage (subquadratic or near-linear for any $\varepsilon>0$) and query time that remains sublinear for any sublinear $t_q$. We supplement our worst-case bounds with a query-driven preprocessing algorithm to build data structures that are well-adapted to the query distribution. 
\end{abstract}

\section{Introduction}

Range searching is a fundamental problem 
with 
wide applications in areas such as computational geometry~\cite{10.5555/1370949}, databases~\cite{10.5555/993519}, geographic information systems (GIS)~\cite{70abcc1cc0244632a97f01a50f4ee449}, and computer graphics~\cite{RTR4}. One fundamental variant of range searching is that of \emph{range counting}, where given a dataset $P$, we need to preprocess it into a data structure so that for a query range $R_q$, we can efficiently estimate $|R_q\cap P|$. Primary examples of types of query ranges include rectangles, balls, halfspaces, and simplices.

In this paper, we study the special case of range counting queries for Euclidean balls of fixed radius. Solving the problem exactly seems to require space usage that is exponential to the dimension, a phenomenon known as the curse of dimensionality. Therefore, aiming for solutions that scale well with the dimension, we focus on approximate solutions where points up to $(1+\eps)r$ away from $q$ may be considered, where $\eps>0$ is an input parameter known during preprocessing. 
Formally,  given as input a dataset $P \subseteq \RR^d$, such that each point $p\in P$ is assigned a weight $w_p \in \RR$,  search radius $r>0$, and approximation parameter $\eps>0$,  the \emph{approximate spherical range counting} problem is to preprocess $P$ into a data structure that given any query point $q\in \RR^d$, it reports 
    $w_S = \sum_{p\in S} w_p$, where $S \subseteq P$ is  
    a subset such that 
    $B (q,r)\cap P\subseteq S \subseteq B(q,(1+\eps)r)\cap P$, where, for any $x\in \RR^d$, $r>0$, $B(x,r)$ denotes the Euclidean ball of radius $r$ centered at $x$.
    %$B(q,r)\cap P\subseteq S \subseteq B(q,(1+\eps)r)\cap P$. 
Since we can uniformly rescale the dataset, we assume throughout the paper that the search radius is $1$.

\subparagraph*{Related work} We juxtapose the above problem definition with a different notion of approximate counting, adapted for example by Aronov and Har-Peled~\cite{AH08}, where one is interested in approximating the number of points inside a query range. Given a query range $R_q$, the data structure shall return a value $\nu$ such that $(1-\eps)|R_q\cap P|\leq \nu\leq (1+\eps)|R_q\cap P|$. Aronov and Har-Peled~\cite{AH08} show that one can reduce such a counting
query to a sequence of ``emptiness" queries, determining if the query range is empty or not, of length polynomial in $\log |P|$ and $\eps^{-1}$. In contrast, our definition asks for an \emph{exact} aggregate over a set whose geometry is allowed to vary within the annulus $B(q,1+\eps)\setminus B(q,1)$. This definition has been used, e.g., by Arya and Mount~\cite{AM00}, although they study solutions in the low-dimensional regime.

In terms of known solutions, approximate spherical range counting in high dimensions is 
not clearly distinguished from other variants of approximate range searching. Perhaps the most promising approach for related problems in high dimensions is locality sensitive hashing (LSH) (see \cite{IM98, HIM12, AI08, ALRW17}), which maps the points to buckets with the promise that near (resp.~far) points are more (resp.~less) likely to be mapped together. 
It is not clear how one could use an LSH-based approach to solve the counting variant without paying at least the cost of exploring all points in $B(q,1)$. This suggests that, within LSH-based frameworks, counting is unlikely to admit a fundamentally more efficient solution than reporting. 
In this work, we focus on designing solutions better suited to the counting variant.

% \begin{problem}[Approximate Spherical Range Counting]
%     Given as input a dataset $P \subseteq \RR^d$, such that each point $p\in P$ is assigned a weight $w_p \in \RR$, and parameters $r>0$, $\eps>0$,  the approximate spherical range counting problem is to preprocess $P$ into a data structure that given any query point $q\in \RR^d$, it reports 
%     $w_S = \sum_{p\in S} w_p$, where $S \subseteq P$ is  
%     a subset such that $B(q,r)\cap P\subseteq S \subseteq B(q,(1+\eps)r)\cap P$. 
% \end{problem}

\subparagraph*{Our Results}
Our main contribution is to design a data structure for the approximate range counting problem in high-dimensional Euclidean spaces. For any constant approximation parameter $\eps \in (0,1)$, our data structure uses near-linear space, polynomial preprocessing time, and query time that is sublinear for any sublinear number of points in the ambiguity zone, i.e., for any sublinear $t_q = |\left(B(q,1+\eps)\setminus B(q,1)\right)\cap P |$ the query time is sublinear. 

Our approach is reminiscent of classic results for the \emph{exact} range searching problem in low-dimensional Euclidean spaces. In particular, a central notion of our approach is that of partition trees~\cite{M92,chan2010optimal}. We adapt well-studied relations between spanning trees with low stabbing\footnote{which also appears as ``crossing" in previous works} numbers and partition trees ~(see, e.g.,~\cite{CW89}) to our approximate setting. A key difference in our setting is that we use a stronger notion of stabbing; we say that a point $q$ $\eps$-stabs a set $S$ if $\exists x
\in S: \|x-q\|\leq 1$ and $\exists y
\in S: \|y-q\|\geq 1+\eps$ (classic stabbing would correspond to the case where $\eps=0$). 
An important step in our preprocessing algorithm is finding a ``light" edge, i.e., a pair of input points that is $\eps$-stabbed by a \emph{polynomial in $d$} ratio of the relevant queries.  
Then, using the multiplicative weight update (MWU) method in a similar manner as~\cite{W89,CW89,W92}, we are able to show that we can compute a spanning tree with low $\eps$-stabbing number and logarithmic maximum degree, implying that we can build a partition tree such that for any query, we only need to visit a bounded number of nodes. To traverse the partition tree, we need to be able to decide for each internal node if the corresponding subset is $\eps$-stabbed by the query. We can only do this approximately, meaning that we may also explore nodes that are almost $\eps$-stabbed, but the overall cost of this exploration is proportional to $t_q$.

\begin{theorem}[Simplified version of  \Cref{thm:final}] Given a set of $n$ points in $\RR^d$, we can build a randomized data structure for the approximate range counting problem, with space usage in $\tilde{{O}}(n)$, preprocessing time in $O(dn)+n^{\mathsf{poly}(1/\eps)}$
%$\tilde{O}(n)^{O(\eps^{-2}\log(1/\eps))}$
and sublinear query time if $t_q$ is also sublinear.
%query time in     $\tilde{O}(dn^{1-\rho^2} + t_q^{\rho}\cdot n^{1-\rho})$ time, where $\rho\in \Theta(\eps^{2}/\log(1/\eps))$.
\end{theorem}
Our data structure is the first to achieve space usage that is near-linear in $n$ and polynomial in $d$, and query time that is sublinear in $n$ and polynomial in $d$ for any sublinear $t_q$. In \Cref{table}, we juxtapose our result with previous known results that even implicitly solve the same problem. While LSH is originally proposed for the approximate near neighbor problem, where the data structure is only required to report one witness point within distance $1+\eps$, Ahle et al.~\cite{AAP17} proposed an LSH-based solution for the problem of reporting all points. Their data structure is adaptive; while building time and space can be that of an LSH table prepared for worst-case inputs (parameter $\rho$ in \Cref{table}), query time is that of the best LSH parameters for the query/instance (parameter $b$).  
Moreover, LSH tuned for near-linear space~\cite{ALRW17} can solve the reporting problem with only a logarithmic overhead in space and query time, since we can amplify the probability that every near neighbor will eventually be encountered. Since both LSH-based methods can solve the reporting variant, they can also solve the counting problem. However, it is not clear whether one can modify such an approach to avoid looking into all near points. 
Although our result does not guarantee sublinear query time in the worst case, it advances prior space-efficient work in high dimensions by making the complexity depend on the number of points in the annulus rather than on all points inside the entire ball. 
Trivially, the bucketing method, when combined with a Johnson-Lindenstrauss dimension reduction (as in \cite{HIM12}), can solve the counting problem; after reducing the dimension, one can store the relevant information in grid cells near the dataset. Such methods that discretize the query space, while providing efficient query times, unavoidably have huge space requirements. Finally, when the dimension is constant, partition trees~\cite{chan2010optimal} solve the problem exactly, while BBD-trees~\cite{AM00} leverage improved query times for approximate solutions (although still exponential in $d$).

In addition to our worst-case guarantees, we present an algorithm inspired by learning theory. Specifically, this algorithm is query-driven, i.e.,~it receives a number of query samples and aims to construct the data structure that yields the best expected query runtime for the distribution of the queries. This way, we design a much faster preprocessing algorithm (from $n^{\mathsf{poly}(1/\eps)}$, the complexity falls to $n^{O(1)}$), which gives us similar bounds for the average case.

Our approach is motivated by the viewpoint of analyzing algorithms through the lens of statistical learning (e.g.,~\cite{GR17, DBLP:journals/jacm/BalcanDDKSV24}). While this has been influential in optimization and mechanism design, to the best of our knowledge, a similar treatment for data structures is absent. Although related ideas appear in learned index structures~\cite{LLZSC25}, LSH trees whose cutting rules are optimized with respect to the given dataset~\cite{AB22},  learning-augmented binary trees with optimized search assuming access to advice from a frequency estimation oracle~\cite{LLW22}, these works do not directly cast the data structure construction itself as a PAC learning problem.
We also remark that the related notions of self-improving~\cite{ACCLMS11} and instance-optimal algorithms~\cite{ABC17} address different guarantees. 
A simplified version of our result follows:

\begin{theorem}[Simplified version of  \Cref{thm:samplingcomplexity}]
    Given a set of $n$ points in $\RR^d$ and a sample of $O(nd\log n)$ queries from a query distribution $\cD_Q$, we can compute a partition tree $\tilde{\ptree}$ with near-optimal expected visiting number for $\cD_Q$ in $O(n^3 d\log n)$ time with high probability. 
\end{theorem}

We believe this is a neat way to use tools from other areas to design efficient algorithms for difficult problems or problems with only existence guarantees, and it would be interesting to see what other data-structure problems might benefit from this viewpoint.

%and encourage the reader to explore this learning theory path in their own research.

\subparagraph*{Organization} In \Cref{sec:prel} we present notation and useful known results and prove auxiliary lemmas that will be used later. In \Cref{section:stabbing} we analyze a data structure that supports approximate stabbing queries. Subsequently, in \Cref{sec:spatreestab} we prove that a spanning tree with sublinear $\eps$-stabbing number exists and can be computed in polynomial time, which we then use to prove that a partition tree with sublinear query time can be efficiently constructed in \Cref{sec:sphrancou}. Finally, in \Cref{sec:datadr} we present our data-driven algorithm for building a partition tree with near-optimal expected visiting number.

\begin{table}[h]
\caption{Existing solutions for the approximate range counting problem. $k$ denotes the number of points within distance $1+\eps$ from the query, whereas $t\leq k$ denotes the number of points at distance between $1$ and $1+\eps$ from the query. 
Adaptive LSH uses a worst-case choice of LSH exponent $\rho$, but answers queries with an instance-best LSH exponent $b$. 
}\label{table}
\centering
\begin{edtable}{tabular}{cccc}
\hline
\textbf{Method} & \textbf{Space} & \textbf{Prep.~time} & \textbf{Query time} \\\hline%& \textbf{Comment} \\\hline
LSH~\cite{ALRW17} & $n^{1+o(1)}$ & $n^{1+o(1)}$ & $n^{1-\Theta(\eps^2)+o(1)}+k n^{o(1)}$ \\%& {tuned for space efficiency} \\ 
Adaptive LSH~\cite{AAP17} & $\tilde{O}(n^{1+\rho})$ & $\tilde{O}(n^{1+\rho})$ & $O(k(n/k)^{b})$ \\%& adaptive $b\leq \rho =\rho(\eps)<1$\\
Bucketing + JL~\cite{HIM12} & $n^{\poly(1/\eps)}$ & $n^{\poly(1/\eps)}$ & $O(\eps^{-2}\log n)$ \\%& -- \\
BBD-Trees~\cite{AM00} & $O(dn)$ & $O(dn\log n)$ & $O(\log n)+O(1/\eps)^{d-1}$ \\%& various range spaces \\
Partition trees~\cite{chan2010optimal} & $O(dn)$ & $O(dn)$ & $O(n^{1-1/d})$ \\%& various range spaces \\
\Cref{thm:final} & $\tilde{O}(n)$ & $n^{\poly(1/\eps)}$ & $ n^{1-\Theta(\eps^4/\log(1/\eps))}$\newline$+t^{\Theta(\eps^2)}\cdot n^{1-\Theta(\eps^2)}$  \\%& $\rho\in \Theta(\eps^{2}/\log(1/\eps))$ \\
\hline
\end{edtable}
\end{table}

\section{Preliminaries}\label{sec:prel}
%A function $f(x_1,\ldots,x_{\ell})$ is in $\poly(x_1,\ldots,x_{\ell})$ if there exist $u_1,\ldots,u_{\ell}$ such that $f(x_1,\ldots,x_{\ell}) \in O\left(\prod_{i=1}^{\ell} x_i^{u_i}\right)$.  
We use the $\tilde{O}$ notation to hide the dependence on $\poly(\log n,\log(1/\eps))$. 
For any $\delta \in \RR_{>0}$, let  $\cG_{\delta}^d=\{\delta \cdot v \mid v\in \ZZ^d\}$ be the grid of side length $\delta$. Unless stated otherwise, $\|\cdot\|$ denotes the Euclidean norm $\|\cdot\|_2$. For any $p\in \RR^d$ and $r>0$, $B(p,r)$ denotes the Euclidean ball centered at $p$, i.e., $B(p,r)=\{x\in \RR^d \mid \|x-p \|\leq r\}$.

We say that a set of points $P$ is  \emph{$(\eps,r)$-stabbed} by a point $p$, if and only if there exists a point $x\in P$ such that $ \|x-p\|\geq (1+\eps)r$ and there exists a point $y\in P$ such that $ \|y-p\|\leq r$. We say that a pair of points $x,y$ (or an edge formed by them) is $(\eps,r)$-stabbed by $p$, if and only if $\{x,y\}$ is $(\eps,r)$-stabbed by $p$. For brevity, when $r=1$, we say that a pair $x,y$ is $\eps$-stabbed by $p$ if and only if it is $(\eps,1)$-stabbed by $p$. Given some $P\subset \RR^d$ and an edge set $G$ (if $G$ is a graph, we refer to its edges), $\sigma(q,G)$ denotes the \emph{$\eps$-stabbing number} of $q$ on $G$, i.e., the number of edges in $G$ that are $\eps$-stabbed by $q$.

For a query $q\in \RR^d$, we partition $P$ into the sets $P_q^{\mathrm{near}}=\{x\in P\mid \|x-q\|\leq 1\}$, $P_q^{\mathrm{far}}=\{x\in P\mid \|x-q\|\geq 1+\eps\}$, and $A_q=P\setminus (P_q^{\mathrm{near}}\cup P_q^{\mathrm{far}})$; notice that $A_q\subseteq B(q,1+\eps)\setminus B(q,1)$. A pair of points of $P$ is $\eps$-stabbed by $q$ if and only if one of them is in $P_q^{\mathrm{near}}$ and the other one is in $P_q^{\mathrm{far}}$. We say that a pair of points of $P$ (or an edge formed by them) is \emph{weakly $\eps$-stabbed} by $q$ if the two points belong to different sets among $P_q^{\mathrm{near}}$, $A_q$, and $P_q^{\mathrm{far}}$, and we denote by $\tau(q,G)$ the number of edges in $G$ that are weakly $\eps$-stabbed by $q$. Clearly, $\sigma(q,G)\leq \tau(q,G)$.

A  \emph{partition tree} $\ptree$ built on a set of points $P$ is a tree where each node $v$ is associated with a subset $P_v \subseteq P$, and for each node $v$ with children $u_1,\ldots,u_{\ell}$, we have $P_v=\bigcup_i P_{u_i}$ and $\forall i\neq j~P_{u_i} \cap P_{u_j} = \emptyset$. %We say that a query range $q$ stabs a set $S$ if $S \cap q \neq \emptyset$ and $S \cap q \neq S$. 
For a partition tree $\ptree$, and a query $q$, the \emph{visiting number} of $\ptree$ w.r.t.~$q$, denoted by $\zeta(q,\ptree)$ is the number of nodes $v$ in $\ptree$ 
which are either root of $\ptree$ or their parent $u$
satisfies one of the following:
\begin{itemize}
    \item $P_u$ is $\eps$-stabbed by $q$, or
    \item $P_u\cap \left(B(q,1+\eps) \setminus B(q,1) \right)\neq \emptyset$ and $P_u \subseteq B(q,1+\eps)$, or
    \item $P_u\cap \left(B(q,1+\eps) \setminus B(q,1) \right)\neq \emptyset$ and $P_u \cap B(q,1)=\emptyset$.
\end{itemize} 
Intuitively, the visiting number is an upper-bound estimate of the number of nodes that will be visited to answer a query approximately.

\subsection{Range spaces}
Fix a countably infinite domain set $X$.  The so-called Vapnik-Chernovenkis dimension~\cite{Sau72, She72, VC71} (VC dimension) of a set $\cF$  
of functions from $X$ to $\{0, 1\}$, denoted by $VCdim(\cF)$, is
the largest $d$ such that there is a set $S=\{x_1, \ldots,x_d\}\subseteq X$, for which, for any $s\in 2^S$ there is a function $f\in \cF$ such that $p\in s \iff f(p)=1$. 
A \emph{range space} is defined as a 
pair $(X,\cF)$ where $X$ is the domain and $\cF$ is a set of functions from $X$ to $\{0, 1\}$, or equivalently a 
pair of sets $(X,\cR)$, where $X$ is the domain and $\cR$ is a set of ranges, i.e., subsets of $X$. The dual range space of a range space defined by sets $(X,\cR)$ is a range space defined by sets $(\cR, \cX)$, where $\cX = \{\{R\in \cR \mid x\in R\} \mid x\in X\}$. The dual VC dimension of a range space is the VC dimension of its dual. A set $N\subseteq X$ is an $\eps$-net if
$\forall R  \in \mathcal{R}: 
\; |R\cap X| \ge \varepsilon |X| \;\Rightarrow\; N \cap R \neq \emptyset$. 
\begin{theorem}[$\eps$-net theorem~\cite{HW87}]
\label{thmepsnet}
Let $(X,\mathcal{R})$ be a range space of VC-dimension $D < \infty$,  
where $X$ is a finite set.  
 There exists a constant $C>0$ such that for any $\eps \in (0,1)$ if we sample a set of $m$ points 
(independent picks with uniform probability) from $X$, where 
\[
m \;\ge\; 
C \,\frac{1}{\varepsilon}
\Bigl(
D \,\log \tfrac{1}{\varepsilon}
\;+\;
\log \tfrac{1}{\delta}
\Bigr),
\]
then, with probability at least $1-\delta$, the resulting set is an $\eps$-net for $(X,\mathcal{R})$.
\end{theorem}

For any $S =(x_1,\ldots,x_m) \in X^m$ let $\mu_{f}(S)=m^{-1}\cdot \sum_{i=1}^m f(x_i)$. For any $\eta>0, r\geq 0,s\geq 0$, let  $d_{\eta}(r,s) = \frac{|r-s|}{\eta+r+s}$. By properly setting $\eps$ in the following theorem from 
\cite{10.5555/338219.338267}, we obtain \Cref{theo:etaepsilonapproxfunctionsrelative}. 

\begin{theorem}[Restated from {\cite{10.5555/338219.338267}}]
\label{theo:etaepsilonapproxfunctions}
    Let $X$ be a countably infinite domain. Let $\cF$ be a set of functions from $X$ to $\{0,1\}$ with VC dimension $D$ and let $\cD$ be a probability distribution over $X$. Also, let $\eps, \delta, \eta \in (0,1)$. There is a universal constant $c$ such that if $m \geq \frac{c}{\eta \cdot \eps^2} \cdot \left( D \log \left(\frac{1}{\eta}\right) + \log\left(\frac{1}{\delta}\right) \right)$
    then if we sample $S = \{x_1,\ldots,x_m\}$  independently at
random according to $\cD$, then with probability at least $1 - \delta$, for any $f\in \cF$, 
    $d_{\eta}\left( \mu_f(S), \EE_{p\sim\cD}[f(p)] \right)\leq \eps.$ 
\end{theorem}

A result similar to \Cref{theo:etaepsilonapproxfunctionsrelative} is obtained in \cite{HS11}. Due to minor differences in the statement, we include a proof (diverted to \Cref{app1}).

\begin{restatable}{corollary}{corapproxfunctions}
\label{theo:etaepsilonapproxfunctionsrelative}
    Let $X$ be a countably infinite domain. Let $\cF$ be a set of functions from $X$ to $\{0,1\}$ with VC dimension $D$ and let $\cD$ be a probability distribution over $X$. Also, let $ \delta, \eta \in (0,1)$. There is a universal constant $c$ such that if $m \geq \frac{c}{\eta } \cdot \left( D \log \left(\frac{1}{\eta}\right) + \log\left(\frac{1}{\delta}\right) \right)$
    then if we sample $S= \{x_1,\ldots,x_m\}$  independently at
random according to $\cD$, then with probability at least $1 - \delta$, for any $f\in \cF$, 
\[
 \left\lvert   \mu_f(S)-\EE_{p\sim\cD}[f(p)]\right\rvert \leq (3/8) \cdot \max\{\EE_{p\sim\cD}[f(p)] , \eta\}.
\]
\end{restatable}

A key requirement in building our data structures is to be able to efficiently sample from a multiset, which can also be updated. Note that a multiset can be seen as a weighted set where weights correspond to multiplicities. In the following lemma, we show that one can maintain a multiset in a data structure that supports sampling and updates on the multiplicities. The proof is diverted to \Cref{app2}.  
\begin{restatable}{lemma}{lemmamultisampling}[Multiset sampling data structure]
\label{lemma:multisampling}
Let $P=\{p_1,\ldots,p_n\}\subset \RR^d$, and 
 let $\eps\in (0,1)$. 
 Let $\Gamma=\{\gamma_1,\ldots,\gamma_n\} \subset \NN$ be weights associated with the points of $P$.  
 There is an algorithm that preprocesses $\Gamma, P$ into a data structure, in $O(n\log n)$ time, using $O(n\log n)$ space, which supports the following operations:
\begin{enumerate}[label=\roman*)]
    \item sampling of a point $p_i\in P$ with probability $\frac{\gamma_i}{\sum_{i\in [n]}\gamma_i}$ in $O(\log(n))$ time, \item update of the weight $\gamma_i$ of a point $p_i\in  P$, in $O(\log n)$ time.
\end{enumerate}    
Sampling uses $O\left(\log \left(\sum_{i\in [n]}\gamma_i\right)\right)$ random bits. If bit-complexity were taken into account, both preprocessing and updates would incur an additional $O\left(\log \left(\sum_{i\in [n]}\gamma_i\right)\right)$ bit-cost per arithmetic operation. 
\end{restatable}

\subsection{Embeddings}
We use the following version of the Johnson-Lindenstrauss (JL) lemma~\cite{JL84}, as expressed for a single vector and Gaussian matrices. 
\begin{theorem}[\cite{IN07}]
\label{lem:JL}
Let $x\in\mathbb{R}^d$ be fixed and let $k\in\mathbb{N}$.
Draw a random matrix $A\in\mathbb{R}^{k\times d}$ with i.i.d.\ entries
$A_{ij}\sim \mathcal{N}(0,1/k)$.
Then for any $\varepsilon\in(0,1)$,
\[
\Pr\!\left[
  \bigl|\|Ax\|_2^2-\|x\|_2^2\bigr|
  > \varepsilon\,\|x\|_2^2
\right]
\;\le\;
2\cdot \exp\!\left(-\frac{\eps^2 k}{8}\right).
\]
\end{theorem}
We also make use of the following result on terminal embedding for Euclidean spaces.
\begin{theorem}[\cite{NN19}]
\label{thm:terminals}
Let $T \subset \RR^d$ be a finite set of 
\emph{terminals}. For every $\eps \in (0,1)$ there exists a mapping 
$f : \RR^d \to \RR^k$ with 
$k = O\!\left(\eps^{-2} \log |T|\right)$ 
such that for all $q \in T$ and all $p \in \RR^d$,
\[
(1 - \eps)\, \|q-p\|_2 
\;\le\; \|f(q) - f(p)\|_2 
\;\le\; (1 + \eps)\, \|q-p\|_2.
\]
For any $T$, and any $c>0$, we can compute a mapping $f$ in $\poly(|T|)$ time such that the above property holds with probability at least $1-n^{-c}$. Moreover, for any $p\in \RR^d$, $f(p)$ can be computed in $\poly(|T|)$ time. 
\end{theorem}
\subsection{Model}
\label{sec:model}
We assume that our algorithms operate in the Real-RAM model with the additional assumption that the floor function can be computed in constant time. 
While it is known that this model is strictly more powerful than the classical word‑RAM~\cite{C06}, it is often used in the Computational Geometry literature~\cite{H11}. We stress that we do not use encodings that pack superpolynomial information into a single machine word or real number, and our use of $\lfloor\cdot\rfloor$ is purely for convenience (e.g., discretization and bucketing).

One building block of our algorithms is a data structure known as the dictionary. 
A \emph{dictionary} is a data structure which stores a set of (key, value) pairs
and when presented with a key, either returns the corresponding value, or returns that the key is
not stored in the dictionary. In our case, keys are typically vectors in $\mathbb{Z}^d$ and values are pointers to lists. 
In the Real-RAM model, we can use a balanced binary tree to build a dictionary with $n$ values, using $\tilde{O}(dn)$ preprocessing time and storage and $\tilde{O}(d)$ query time.  
\section{Approximate Stabbing Queries}
\label{section:stabbing}
In this section, we design a solution for our main subproblem, that is, the problem of preprocessing a set of points to support queries that ask to determine if a query $\eps$-stabs the given pointset. In particular, the data structure is required to determine if the query $\eps$-stabs the pointset, if all points are near the query, or if all points are far from the query. The resulting data structure returns approximate answers, meaning that it allows for wrong decisions when the pointset is ``nearly" $\eps$-stabbed, as formalized in \Cref{thminternalnode}. This is sufficient for our purposes, which are building these data structures in internal nodes of a partition tree to assist in traversing the tree to answer a query.

Our main technical ingredient is a randomized embedding to the Hamming metric, which is implemented using LSH to partition the space into buckets and then applying random binarization of the buckets. A similar argument has been used in Alman et al.~\cite{ACW16}, where they first embed the points to $\ell_1$ and then use an LSH function to partition points into buckets; each bucket is then assigned to a random bit. Repeating the above sufficiently many times gives the desired embedding. We obtain a slightly more direct construction that does not need to embed the points in $\ell_1$, but instead uses the LSH family of Datar et al.~\cite{DIIM04}. The final embedding succeeds in preserving distances with respect to a single scale; distances considered small with respect to a given threshold remain small, and large distances remain large. This is formally stated in the following lemma. 
%Its proof is diverted to \Cref{app:embedtohamming}.

\begin{restatable}{lemma}{lemmaembedtohamming}
\label{lemma:embedtohamming}
For any $\eps \in (0,1)$ and any integer $d'\geq 1$, there exists a randomized embedding $f:~\RR^d \to \{0,1\}^{d'}$, and a specific threshold $\theta\in [0,d']$, such that, for any $p,q \in \RR^d$,
\begin{itemize}
    \item if $\|p-q\|\leq 1$ then $\Pr\left[\|f(p)-f(q)\|_1 \geq \theta\right]\leq e^{-\frac{\eps^2 d'}{115200}}$, and
    \item if $\|p-q\|\geq 1+\eps$ then $\Pr\left[\|f(p)-f(q)\|_1 \leq \theta\right]\leq e^{-\frac{\eps^2 d'}{115200}}$.
\end{itemize}
Moreover, for any $p\in \RR^d$, $f(p)$ can be computed in $\tilde{O}(dd')$ running time. 
\end{restatable}
\begin{proof}
%We assume wlog that $r=1$. 
We employ the LSH family of \cite{DIIM04}. For $p \in {\RR}^d$, consider the random function $h(p)=\left \lfloor \frac{\langle p, g\rangle+s}{w} \right \rfloor $,
where $g$ is a $d$-dimensional vector with i.i.d.~random variables following $N(0,1)$, $s$ is a shift chosen uniformly at random from $[0,w]$, and $w$ is a fixed parameter defining the side-length of the one-dimensional grid on the line. We set $w=1+\eps$. Then, the probability of collision for two points $p,q$ is 
\[
\int_{s=0}^{1+\eps} \frac{2}{\sqrt{2\pi}\|p-q\|}e^{\left(-\frac{s^2}{2{\|p-q\|}^2}\right)}\left(1-\frac{s}{1+\eps}\right) ds. 
\]
Let $p_1 := \Pr[h(p)=h(q)\mid \|p-q\|\leq 1]$ and $p_2 := \Pr[h(p)=h(q)\mid \|p-q\|\geq 1+\eps]$. 
Using straightforward calculus, we obtain the following:
\begin{equation}\label{eq:pr}
    p_1-p_2 \geq \eps/5 \text{ and } p_1 \leq 2/3.
\end{equation}

We now sample $d'$ independent hash functions from the above family. 
For each $i\in[d']$, we denote by $h_i(P)$ the image of $P$ under $h_i$, i.e., a set of nonempty buckets.
Now each point in each nonempty bucket $b\in h_i(P)$ is mapped to $\{ 0,1 \}$: with probability $1/2$, for each $x\in b$, set $f_i(x)=0$,  otherwise set $f_i(x)=1$. 
We define $f(p)=(f_1(h_1(p)),\ldots,f_{d'}(h_{d'}(p))).$ 
Now, observe that 
\begin{align*}
    \|p-q\|\leq 1 \implies \mathbb{E}\left[\|f(p)-f(q)\|_1 \right] = \sum_{i=1}^{d'} \mathbb{E}\left[ |f_i(h_i(p))-f_i(h_i(q))|\right] 
    \leq  0.5 d' \left(1- p_1\right)=:\mu_1,
\end{align*}
and
\begin{align*}
    \|p-q\|\geq 1+\eps \implies \mathbb{E}\left[\|f(p)-f(q)\|_1 \right] = \sum_{i=1}^{d'} \mathbb{E}\left[ |f_i(h_i(p))-f_i(h_i(q))|\right] 
    \geq  0.5 d' \left(1- p_2\right)=:\mu_2. 
\end{align*}
Now, let $\theta=\mu_1(1+\eps/80)$, we conclude the statement. By (\ref{eq:pr}), for any $\eps\in (0,1)$, 
\[
\theta = \mu_1(1+\eps/80) <
\mu_2(1-\eps/80),
\]
and, since $p_2<p_1\leq 2/3$, we have $\mu_2\geq \mu_1\geq d'/6$. Notice that $\|f(p)-f(q)\|_1$ is a sum of $d'$ independent random variables taking values in $\{0,1\}$. Hence, by Chernoff bounds, if $\|p-q\|\leq 1$, then
\[
 \Pr\left[\|f(p)-f(q)\|_1\geq (1+\eps/80)\mu_1 \right]\leq exp\left(-\frac{(\eps/80)^2\mu_1}{3}\right)\leq exp\left(-\frac{\eps^2 d'}{115200}\right)
\]
and, if $\|p-q\|\geq 1+\eps$, then
\[
\begin{split}
 \Pr\left[\|f(p)-f(q)\|_1 \leq (1+\eps/80)\mu_1 \right]&\leq \Pr\left[\|f(p)-f(q)\|_1\leq (1-\eps/80)\mu_2 \right]\\
 &\leq exp\left(-\frac{(\eps/80)^2\mu_2}{2}\right)\leq exp\left(-\frac{\eps^2 d'}{115200}\right).
\end{split}
\]
\end{proof}

We aim to build a data structure that supports approximate ``stabbing queries". Given a query point $q$ we need to identify if there exists a pair of points in the input pointset that are (approximately) $\eps$-stabbed by $q$. 
Using the embedding provided by \Cref{lemma:embedtohamming}, we map points to keys in $\{0,1\}^{d'}$, and we store them in a dictionary, where each bucket is assigned the list of points mapped to it. We set $d'$ slightly below $\log n$ to balance the number of buckets $2^{d'}$ with the number of arbitrarily distorted distances from $q$, in a way that both remain sublinear. To answer a query, we search the dictionary for an approximate near neighbor, i.e., a point within distance $1+\eps$, and an approximate far neighbor, i.e., a point at a distance of at least $1$. The following lemma indicates that a simple search over the buckets of a single dictionary is sufficient to guarantee that if there is a point within distance $1$, we will find a point within distance $1+\eps$ with constant probability, and if there is a point at distance at least $1+\eps$, we will find a point at distance at least $1$ with constant probability. We remark that for finding an approximate near neighbor, one could instead use the data structure of Andoni et al.~\cite{ALRW17} that provides near-linear space with comparable query time. However, we are not aware of a comparable result for computing an approximate far neighbor (which we also need and do). The complete pseudocode is diverted to \Cref{app:pseudocode}. 
\begin{restatable}{lemma}{thmstabds}
\label{thm:stabds}
Let $P=\{p_1,\ldots,p_n\}\subset \RR^d$ and let $\eps\in (0,1)$. 
    We can build a randomized data structure with preprocessing time and space in $\tilde{O}(dn)$ and query time in $\tilde{O}(dn^{1-\frac{\eps^2}{115200 ( 1+\eps^2)}})$, such that for any query $q\in\RR^d$, it reports $\mathcal{I}$, $|\mathcal{I}|\leq 2$, as follows:
    \begin{itemize}
    \item if $\exists x \in P:~\|x-q\|\leq 1$ then $\Pr[\exists i\in \mathcal{I} \text{ such that }\|q-p_i\|\leq 1+\eps] \geq 0.95$,
    \item if $\exists y \in P:~\|y-q\|\geq 1+\eps$ then $\Pr[\exists j\in \mathcal{I} \text{ such that }\|q-p_j\| \geq 1]\geq 0.95$, 
    \end{itemize}
where the probabilities are taken over the randomness during preprocessing. 
\end{restatable}
\begin{proof}
    \emph{Data Structure.}
We sample an embedding $f$, as in  \Cref{lemma:embedtohamming}, to map points of $P$ to the Hamming metric of dimension $d'=\frac{\log n}{1+\eps^2}$. We store points of $P$ in a dictionary $\cD$; each point $p\in P$ is associated with a key $f(p)\in \{0,1\}^{d'}$. For each non-empty bucket in the dictionary, we store a linked list with all points falling in that bucket. The complete pseudocode of the preprocessing algorithm can be found in \Cref{alg:preprocessing}, in \Cref{app:pseudocode}. Using balanced binary trees to implement $\cD$, the preprocessing algorithm runs in $\tilde{O}(dn)$ time, and the data structure uses $\tilde{O}(dn)$ space, as discussed in \Cref{sec:model}.
%\andreas{explain why for both time and space}.\yiannis{you mean time and space of dictionary? it's discussed in sectin "model" too, but i would assume it's basic..} 
    
    \emph{Query Algorithm.} When a query point $q\in \RR^d$ arrives, we first compute $f(q)$, and we start looking at points in buckets of $\cD$ whose key is of Hamming distance at most $\theta$ (as in \Cref{lemma:embedtohamming}) from $f(q)$. We stop when either we find a point whose pre-image is within distance $1+\eps$ from $q$, and we include it in $\mathcal{I}$, or when we have seen $100 n^{1-\frac{\eps^2}{115200(1+\eps^2)}}$ (irrelevant) points of $P$. Then, we start looking at points in buckets of $\cD$ whose key is of Hamming distance at least $\theta$ (as in \Cref{lemma:embedtohamming}) from $f(q)$. We stop when either we find a point whose pre-image is at a distance of at least $1$ from $q$, and we include it in $\mathcal{I}$, or when we have seen $100 n^{1-\frac{\eps^2}{115200(1+\eps^2)}}$ (irrelevant) points of $P$. The complete description of the algorithm can be found in \Cref{alg:query}, in \Cref{app:pseudocode}. The query time is upper bounded by the number of available buckets, i.e., $O(n^{1-\frac{\eps^2}{(1+\eps^2)}})$, plus the number of irrelevant candidates checked through the above iteration, i.e., $O(n^{1-\frac{\eps^2}{115200(1+\eps^2)}})$. 

\emph{Correctness.}
    Let $N_{far} = |\{p\in P\mid \|p-q\| \geq 1+\eps \text{ and } \|f(p)-f(q)\|_1\leq \theta\}|$ be the number of points of $P$ which are at a distance of at least $1+\eps$ from $q$ but are embedded to within a distance $\theta$ from $f(q)$. By \Cref{lemma:embedtohamming}, $\EE[N_{far}]\leq n\cdot \Pr[\|f(p)-f(q)\|_1\leq\theta\ |\ \|p-q\|\geq 1+\eps]\leq n\cdot e^{-\frac{\eps^2 d'}{115200}}\leq n^{1-\frac{\eps^2}{115200(1+\eps^2)}}$, as $d'=\frac{\log n}{1+\eps^2}$. By Markov's inequality $\Pr[N_{far}\geq 99 n^{1-\frac{\eps^2}{115200(1+\eps^2)}}]\leq 1/99$. Moreover, the probability that a point within distance $1$ from $q$ is embedded to a distance larger than $\theta$ from $f(q)$ is at most $n^{-\frac{\eps^2}{115200(1+\eps^2)}}$. By a union bound, any of the above two events happens with probability at most $1/99+ n^{-\frac{\eps^2}{115200(1+\eps^2)}} \leq 2/99<0.05$. 
    Hence, if there is a point of $P$ within distance $1$ from $q$, the preprocessing will be successful for $q$, with probability at least $0.95$, in that the query algorithm will find a point within distance $1+\eps$ from $q$.  
    Similarly, the expected number $\EE[N_{near}]$ of  points of $P$ which are within distance $1$ from $q$ but are embedded to a distance of at least $\theta$ from $f(q)$ is at most $n^{1-\frac{\eps^2}{115200(1+\eps^2)}}$. By a union bound and Markov's inequality, the probability that a far neighbor of $q$ (a point at distance at least $1+\eps$) is mapped correctly under $f$, and that $N_{near}< 99 n^{1-\frac{\eps^2}{115200(1+\eps^2)}}$ is at least $0.95$. %Hence, by a union bound, the overall probability of success is at least $0.9$\andreas{this is not the theorem statement, maybe erase last sentence}.
\end{proof}

Now using \Cref{thm:stabds}, we can prove the main result of this section, namely, a data structure that determines, with high probability, if a query point $\eps$-stabs the preprocessed pointset, if it is near the entire pointset {(the pointset is \emph{covered} by $q$)}, or if it is far from the entire pointset {(the pointset is \emph{disjoint} from $q$)}. The data structure gives approximate answers in the following sense: if all points are within distance $1+\eps$, and there is at least one point at a distance at most $1$, it can falsely report that the pointset is $\eps$-stabbed. Similarly, if all points are at a distance larger than $1$, and there is at least one point at a distance of at least $1+\eps$, it can falsely report that the pointset is $\eps$-stabbed. If all points lie in the annulus $\{x\in \RR^d\mid \|x-q\|\in (1,1+\eps)\}$ then any answer is possible. 
Additionally, the success probability of the preprocessing is for each query individually, i.e., it is not for all queries simultaneously. Copies of this structure will be used as parts of the final data structure for the approximate spherical range counting problem.

\begin{theorem}\label{thminternalnode}
Let $P\subset \RR^d$ be a set of $n$ points and let $\eps\in (0,1)$.
We can build a randomized data structure such that for any query $q\in \RR^d$, it reports as follows:
\begin{itemize}
     \item if $\exists x \in P:~\|x-q\|\leq 1$ and $\exists y \in P:~\|y-q\|\geq 1+\eps$ then it returns ``stabbed",
    \item if $\forall y \in P:~\|y-q\|\leq 1$ then it returns ``covered"
    \item  if $\forall x \in P:~\|x-q\| \geq 1+\eps$ then it returns ``disjoint",
    \item if $\exists x \in P:~\|x-q\|\leq 1$ and $\forall y \in P:~\|y-q\|< 1+\eps$  then it can return ``stabbed", or ``covered", and
    \item if $\exists x \in P:~\|x-q\|\geq  1+\eps$ and $\forall y \in P:~\|y-q\|>1$ then it can return ``stabbed", or ``disjoint".
    \item if $\forall x \in P:~\|x-q\|\in (1,1+\eps)$ then it can return ``stabbed", or ``disjoint", or ``covered".
    \end{itemize}

%\andreas{for any or for all?}\yiannis{do you prefer the wording "for a fixed query $q$"?}
For each query and for any constant $c>0$, the preprocessing of the data structure succeeds with probability $1-n^{-c}$,
%For any query $q$, for any constant $c>0$, the preprocessing of the data structure succeeds with probability $1-n^{-c}$,
has $\tilde{O}(dn)$ preprocessing time and space and the query time is $\tilde{O}(dn^{1-\frac{\eps^2}{115200 ( 1+\eps^2)}})$.       
\end{theorem}
\begin{proof}
%We start by describing a data structure that achieves constant probability of success. The final result then follows by independently building $O(\log n)$ data structures. 
We build $L=c\log n$ independent data structures as in \Cref{thm:stabds}. For any query point $q\in \RR^d$, we query the data structures, and we report as follows. If we find a point $x$ such that $\|x-q\|\leq 1+\eps$ and a point $y$ such that $\|y-q\|\geq 1$ we return ``stabbed". Otherwise, if all returned points are within distance $1+\eps$, we return ``covered". Else, we return ``disjoint".

If there is a pair $x,y$ in $P$ such that $\|x-q\|\leq 1$ and $\|y-q\| \geq 1+\eps$ then the probability that all $L$ data structures fail to return a point $x'$ such that $\|x'-q\|\leq 1+\eps$ or a point $y'$ such that $\|y'-q\|\geq 1$ is by union bound at most $2\cdot 0.05^L \leq n^{-c}$. Moreover, if all points of $P$ are within distance $1$ from $q$ (resp.~at a distance of at least $1+\eps)$, then all returned points are within distance $1+\eps$ from $q$ (resp.~at a distance of at least $1$), implying that the query algorithm correctly identifies ``covered" (resp.~``disjoint").  
\end{proof}

\section{Spanning Trees with low $\eps$-stabbing number}\label{sec:spatreestab}

The $\eps$-stabbing number of a spanning tree (or of a set of edges in general) is the maximum number of edges of that tree (or that set) that are $\eps$-stabbed by any query range. Chazelle and Welzl showed in \cite{CW89} a close relationship between the stabbing number of a spanning tree (the stabbing number is the special case of $\eps$-stabbing number where $\eps=0$) and the maximum number of nodes visited, during a worst-case query, in a corresponding partition tree. %(the visiting number is the maximum number of nodes visited - both inner and leaves - during a worst case query). 
Motivated by this, we wish to extend this notion to the $\eps$-stabbing number and to achieve similar results, but without exponential dependency on $d$ as we deal with the approximation scheme. In this section, we show that a spanning tree with sublinear $\eps$-stabbing number exists, and can be computed in polynomial time. 
Following \cite{CW89}, we use MWU to find a spanning tree with small $\eps$-stabbing number. We remark that there are other alternatives~\cite{CM21,H09} in the literature for computing spanning trees with low stabbing numbers. While it is plausible that they can be modified to support our notion of $\eps$-stabbing, they all share the same bottleneck, i.e., they depend at least linearly on the number of ranges.  

%In \Cref{sec:sphrancou} we show how to get a partition tree with sublinear visiting number given a spanning tree with sublinear $\eps$-stabbing number.
%each node is assigned to a data structure for the approximate stabbing problem so that a query can be answered by traversing the tree and querying each data structure along the way.   

Similarly, to \cite{CW89}, we start by showing that for any set of ball queries, we can always find a pair of points that is not $\eps$-stabbed many times.

\begin{lemma}\label{lemma:lightedge}
    Let $P\subset \RR^d$ be a set of $n$ points and let $Q$ be a
multiset consisting of $m$ points from a universe set  $U \subset \RR^d$.  %$\cG_{\eps/\sqrt{d}}$. 
Let $\eps \in (0,1)$. 
There exists a pair of points in $P$ which is not $\eps$-stabbed
by more than $O\left(m \cdot \frac{d}{n^{\rho}}\right)$
%\todo[inline]{I think $\delta = \frac{d}{2^{\Theta(\eps \sqrt{\log n})}}$ which seems bad (?)} 
points of $Q$, where $\rho = \Theta(\eps^2/\log(1/\eps))$.

%\nc{Moreover, if $Q$ is stored in a multiset sampling data structure, we can find such a pair of points in $\tilde{O}(n\cdot \poly(d,\frac{1}{\eps})) \cdot O(\log |U|)$ time with probability $1-1/\poly(n)$.}
If $Q$ is stored in a multiset sampling data structure as in \Cref{lemma:multisampling}, then for any constant $c$, we can find such a pair of points in $\poly(n, d,\frac{1}{\eps}) \cdot O(\log |U|)$ time with probability $1-n^{-c}$.
\end{lemma}
\begin{proof}
If $d\geq n^{\rho}$, then the statement trivially holds, since $|Q|=m$. Thus, we assume $d< n^{\rho}$. 
For any two points $x,y$, let $S_{x,y}$ be the set of points in $\RR^d$ that $\eps$-stab $x,y$. We define a range space $\mathcal{R}= (\RR^d,\{ {S}_{x,y}\mid x,y\in \RR^d, x\neq y\})$. Notice that for any point $p\in \RR^d$, determining if $p\in S_{x,y}$ reduces to evaluating the signs of the following four polynomials of degree $2$ in $d$ variables: $\|p-x\|^2 - 1$, $\|p-x\|^2 - (1+\eps)^2$, $\|p-y\|^2 - 1$, $\|p-y\|^2 - (1+\eps)^2$. This implies (see, e.g., \cite{AB02}[Theorem 8.3]) that the VC dimension of $\mathcal{R}$ is $O(d)$. Hence, by \Cref{thmepsnet}, for any $\delta<1$ and constant $c>0$, 
a random sample $N_{\delta}$ of $O\left(\frac{d}{\delta}\left(\log\left(\frac{1}{\delta}\right)+\log n\right)\right) $ points of $Q$ is a $\delta$-net for $\mathcal{R}$ with probability $1-n^{-c-1}$. Moreover, we can find such a $\delta$-net by randomly sampling from $Q$, using the assumed multiset sampling data structure of \cref{lemma:multisampling}, in $\tilde{O}(\frac{d}{\delta}\log n) \cdot O(\log |U|)$ time. 

By \Cref{thm:terminals},  there exists a terminal embedding $f$ to dimension $d' = O(\eps^{-2}\log |N_{\delta}|)$ such that $\forall p\in N_{\delta},q \in \RR^d$, $\|f(p)-f(q)\| \in (1\pm \eps/20)\|p-q\|$. The time needed to compute $f(P)$ is in $O(n) \cdot \poly(|N_{\delta}|)$ for the probability of success to be at least $1-n^{-c-1}$.

Now let $x,y \in P$ be $\eps$-stabbed by some $q\in Q$, and further assume wlog that $\|x-q\|\leq 1$ and $\|y-q\|\geq 1+\eps$. We obtain $\|f(x)-f(q)\|\leq 1+\eps/20$ and 
$\|f(y)-f(q)\|\geq (1+\eps)(1-\eps/20)\geq  1+3\eps/4$. 
Hence, $f(x),f(y)$ are $(\eps/2, 1+\eps/20)$-stabbed by $f(q)$.  
Moreover, by the triangle inequality, $\|f(x)-f(y)\| \geq \eps/2$. Now consider the grid of side-length $\eps/(4\sqrt{d'})$ in $\RR^{d'}$ and let $U_B$ be the set of all grid cells that intersect with at least one ball of radius $1+\eps$ centered at a point of $f(N_{\delta})$. There are 
\[n_{\eps,d'} = |N_{\delta}| \cdot O(1/\eps)^{d'} = |N_{\delta}| \cdot O(1/\eps)^{O(\eps^{-2}\cdot \log |N_{\delta}|)}=|N_{\delta}|^{O(\eps^{-2}\log (1/\eps))}
\] such cells, where $|N_{\delta}|=O\left(\frac{d}{\delta}\left(\log\left(\frac{1}{\delta}\right)+\log n\right)\right)$.  Now let $\delta =\frac{d}{n^{\rho}}$ where $\rho = \Theta(\eps^2/\log(1/\eps))$ is chosen sufficiently small (in terms of the hidden constants) so that $|N_{\delta}|^{O(\eps^{-2}\log (1/\eps))} \leq \left(n^{\rho} \log n\right)^{O(\eps^{-2}\log (1/\eps))} < n-2$.
By the pigeonhole principle, there is either a pair of points of $P$ mapped to the same cell of $U_B$, under $f$, or there is a pair of points of $P$ that are not mapped to $U_B$. In either case, the two points are not $(\eps/2, 1+\eps/20)$-stabbed by any point of $f(N_{\delta})$. To see this, notice that if the two points are in the same cell, then they are within distance $\eps/4$ from each other, contradicting the fact that they are $(\eps/2, 1+\eps/20)$-stabbed (which would mean they are separated by a distance of at least $\eps/2$). If they are not mapped in $U_B$ then neither of them is within distance $1+\eps/20$ from any point of $f(N_{\delta})$, also implying that they cannot be $(\eps/2, 1+\eps/20)$-stabbed by any point of $f(N_{\delta})$. 
Since the two points are not $(\eps/2, 1+\eps/20)$-stabbed by any point of $f(N_{\delta})$ in the mapped space, their pre-images (under $f$) are not $\eps$-stabbed by any point of $N_{\delta}$. Since $N_{\delta}$ is a $\delta$-net, it must be that at most $\delta m$ points of $Q$ $\eps$-stab those two points.   

We now focus on the running time needed to find a pair of points as above. By \Cref{thm:terminals},  we can map points under $f$ in $\poly(n)$ time. 
To compute the two points, we essentially use the bucketing method (see, e.g.,~\cite{HIM12}) for balls centered at $f(N_{\delta})$. 
We first store all cells of $U_B$ in a dictionary data structure. Then, for each point  $x\in f(P)$, we query the data structure to see if its cell $C_x$ is stored, i.e., $C_x\in U_B$. If $C_x\in U_B$ then we store $x$ in a linked list in the bucket of $C_x$. We can stop and output any two points stored in the same linked list, or any two points whose cells do not belong to $U_B$. Iterating over all points of $f(P)$ costs  $\tilde{O}(dn)$ assuming that the dictionary is implemented using balanced binary trees. The overall probability of success follows from a union bound over failure events of the random sampling for computing $N_{\delta}$ and the failure event of computing the terminal embedding $f$.
\end{proof}

Knowing that a good pair always exists, and that we can find it, we repeat the process $\lfloor n/4\rfloor$ times, each time removing both points of the selected pair, in order to build a good matching, in the sense that it has a small $\eps$-stabbing number. Removing both points, rather than one, allows us to bound the maximum degree of the spanning tree that we eventually build, which is needed in \Cref{sec:sphrancou}. To do that, we use the MWU technique introduced in the context of trees with low stabbing number in \cite{CW89}. The proof is very similar to \cite{CW89}. %and is diverted to \Cref{app:lemma:forest}.% and used for similar procedures in \andreas{cite}.

\begin{restatable}{lemma}{lemmaforest}
\label{lemma:forest}
Let $P\subset \RR^d$ be a set of $n$ points and let $\eps \in (0,1)$. 
There is a matching on $P$ with at least $\lfloor n/4\rfloor$
edges, such that every point of $\cG_{\gamma}^d$ $\eps$-stabs $\tilde{O}(dn^{1-\rho})$ of its edges, where $\gamma = \Theta(\eps/\sqrt{d})$, $\rho = \Theta(\eps^2/\log(1/\eps))$.
For any $c>0$, we can compute such a matching in $ \poly(n,d,\frac{1}{\eps})
 + \tilde{O}(n)\cdot O(1/\eps)^d$ time, with probability of success $1-n^{-c}$.
\end{restatable}
\begin{proof}
We initiate our discussion with a high-level description of our algorithm, which is sufficient to show that a matching with the desired properties exists. We then discuss implementation details and bound the running time of the algorithm.

Let $Q_0\gets \bigcup_{p\in P} B(p,(1+\eps))\cap \cG_{\gamma}^d$. Notice that points of $\cG_{\gamma}^d$ not in $\bigcup_{p\in P} B(p,(1+\eps))$ do not $\eps$-stab any pair of points of $P$.  Using \Cref{lemma:lightedge}, we pick as first edge the edge realized by a pair of points $x,y$ of $P$ which is not $\eps$-stabbed by more than $O(dn^{-\rho}|Q_0|)$ points of $Q_0$. Next, we create a multiset $Q_1$, where we duplicate points of $Q_0$ that $\eps$-stab $x,y$, and we place all remaining points of $Q_0$ intact. We repeat for $P\setminus \{x,y\}$ and $Q_1$.   

There are $\lfloor n/4 \rfloor$ iterations. Since every iteration removes two points, every iteration is applied on a set of at least $n/2$ points of $P$, and hence, by \Cref{lemma:lightedge}, the selected pair is $\eps$-stabbed by $O(dn^{-\rho}|Q_i|)$ points of the current multiset $Q_i$. In the end, $Q_{\lfloor n/4 \rfloor}$ has at most
\[|Q_0| \cdot \left(1+O\left(dn^{-\rho}\right)\right)^{\lfloor n/4\rfloor} \leq |Q_0|\cdot exp\left(O(dn^{1-\rho})\right)\] points and, due to the duplication strategy,  any point in $Q_{\lfloor n/4 \rfloor}$ $\eps$-stabs at most
\[\log |Q_0| + O(dn^{1-\rho})= O(d\log(1/\eps) + \log n +dn^{1-\rho})  \text{ edges.}\] Since the selected pairs are pairwise disjoint, they form a matching with $\lfloor n/4\rfloor$ edges.  

To implement the above scheme, we first store $Q_0$ in a multiset sampling data structure $\cD_s$ as in \Cref{lemma:multisampling}. 
By standard bounds on the number of grid cells intersecting a ball (see, e.g.~\cite{HIM12}), we have $|Q_0|\leq n \cdot O(1/\eps)^d$. By \Cref{lemma:multisampling}, the running time needed to build the multiset sampling data structure is $\tilde{O}( n) \cdot O(1/\eps)^d$. 
% We also store $Q_0$ in a dictionary data structure $\cD_b$, where each cell is associated with the points in $P$ whose unit ball intersects that cell. This is also known as the bucketing method~\cite{HIM12} requiring building time in $O(dn)\cdot O(1/\eps)^d$. 
Now for each $i\in [\lfloor n/4 \rfloor]$, in the $i$th iteration, we detect a pair of points $x,y$ which is not $\eps$-stabbed by more than $O(dn^{-\rho}|Q_i|)$ points of $Q_i$ using $\cD_s$ and the algorithm of \cref{lemma:lightedge} in $\poly(n, d,\frac{1}{\eps})$ time (since the universe set consists of at most $n\cdot O(1/\eps)^d$ points), with probability of success $1-n^{-c-1}$. We now detect all points of $Q_i$ that $\eps$-stab $x,y$; this can be done by enumerating all points in $\left(B(x,(1+\eps))\cup B(y,(1+\eps))\right)\cap \cG_{\gamma}^d$ and excluding points that are near to both $x$ and $y$ in $O(1/\eps)^d$ time. For all points $\eps$-stabbing $x,y$, we update their weights in $\cD_s$, needing in total $O(1/\eps)^d$ time. Hence, the overall running time is in $\poly(n, d,\frac{1}{\eps})
 + \tilde{O}(n)\cdot O(1/\eps)^d$, and the probability of success is $1-n^{-c}$ by taking a union bound over the $\lfloor n/4\rfloor$ applications of the randomized algorithm of \cref{lemma:lightedge}.    
\end{proof}

% \Cref{lemma:lightedge,lemma:forest} prove that a forest with useful properties exists. In the following lemma, we show that we can actually find such a tree efficiently.

Finally, we are ready to show the main result of this section, the tractability of a spanning tree with a sublinear $\eps$-stabbing number and logarithmic maximum degree. This follows from $O(\log n)$ iterative applications of \Cref{lemma:forest}. 
%The complete proof can be found in \Cref{app:thm:treeexists}.

\begin{restatable}{theorem}{thmtreeexists}\label{thm:treeexists}
    Let $P\subset \RR^d$ be a set of $n$ points. 
There is a spanning tree on $P$ with maximum degree $O(\log n)$, such that every point of $\cG_{\gamma}^d$ $\eps$-stabs $\tilde{O}(dn^{1-\rho})$ of its edges, where $\gamma = \Theta(\eps/\sqrt{d})$, $\rho = \Theta(\eps^2/\log(1/\eps))$.

Moreover, for any $c>0$, we can compute such a spanning tree in $\poly(n, d,\frac{1}{\eps})
 + \tilde{O}(n)\cdot O(1/\eps)^d$ time, with probability of success $1-n^{-c}$.
\end{restatable}
\begin{proof}
    We apply \Cref{lemma:forest} with probability of success $1-n^{-c-1}$, iteratively as follows.  We first find a matching with at least $\lfloor n/4\rfloor$ edges such that every point of $\cG_{\gamma}^d$ $\eps$-stabs $\tilde{O}(dn^{1-\rho})$ of its edges; together with the unmatched points, it forms a forest with at most $3n/4+1$ trees. Then, for each tree, we keep an arbitrary representative point, and we apply  \Cref{lemma:forest} on the set of representative points. Every edge of the resulting matching connects two distinct trees, hence we obtain a forest again, and the number of trees decreases by a factor of roughly $3/4$. We repeat; we stop when the number of trees is at most $\tilde{O}(dn^{1-\rho})$, and we connect the remaining trees by a path on their representative points, which adds at most $\tilde{O}(dn^{1-\rho})$ edges. Every point of $\cG_{\gamma}^d$ $\eps$-stabs at most 
    \[\tilde{O}\left(\sum_{j\geq 0} d\left(n\cdot (3/4)^{j}\right)^{1-\rho}\right) + \tilde{O}(dn^{1-\rho}) \subseteq  \tilde{O}(dn^{1-\rho})\text{ edges of the spanning tree.}\]
    Moreover, since the edges added in each iteration form a matching, every point gains at most one edge per iteration, and at most two edges from the final path; hence, the maximum degree of the spanning tree is $O(\log n)$.  
    
    There are at most $O(\log n)$ iterations, hence the running time follows from \Cref{lemma:forest} and the probability of success follows by a union bound over the $O(\log n)$ executions of the algorithm of \Cref{lemma:forest}. 
\end{proof}

\section{Spherical Range Counting}\label{sec:sphrancou}

In this section, we show how to construct a partition tree for the approximate spherical range counting problem with sublinear query time. We first show how to get a partition tree with a sublinear visiting number given a spanning tree with a sublinear $\eps$-stabbing number, by extending the procedure in \cite{CW89} to accommodate the approximation scheme. Subsequently, we incorporate the data structures of \Cref{section:stabbing} in the inner nodes of the partition tree to get sublinear worst-case guarantees. 

We start by proving that having a spanning tree with low $\eps$-stabbing number can give us a spanning path with low $\eps$-stabbing number, which we then use to construct a partition tree with low visiting number. \Cref{lem:treepath,lem:pathtree} are adapted from \cite{CW89} for $\eps$-stabbing numbers, and the bounds become dependent on the ambiguity zone of the query and, in the case of \Cref{lem:treepath}, on the degrees of the spanning tree. 
%For clarity and space efficiency, we move their proofs to \Cref{app:treepath,app:pathtree}.

\begin{restatable}{lemma}{lemtreepath}
\label{lem:treepath}
    Given a spanning tree $T$ on a pointset $P\subset \RR^d$,  
    we can build a spanning path $\ppath$ on $P$, in ${O}(dn)$ time, such that for any point $q$,
    \[
    \sigma(q,\ppath)\leq 2\tau(q,T)
    \quad\text{and}\quad
    \sigma(q,\ppath)\leq 2\sigma(q,T)+\sum_{v\in A_q}\deg_T(v)\leq 2\sigma(q,T)+\Delta(T)\cdot t_q,
    \]
    where $\Delta(T)$ is the maximum degree of $T$ and $t_q = |\left(B(q,1+\eps)\setminus B(q,1)\right)\cap P|$. 
    % with $\varepsilon$-stabbing number $\sigma(T)$, we can build a spanning path $\mathcal{P}$ on the same pointset with $\varepsilon$-stabbing number $\sigma(\mathcal{P})= 2\sigma(T)+k$, where $k$ is the size of the output of the query. I.e. given a spanning tree $T$ that is $\varepsilon$-stabbed at most $\sigma(T)$ times by a query $q$, we can build a spanning path $\mathcal{P}$ that is $\varepsilon$-stabbed at most $2\sigma(T)+k$ by the same query $q$, where $k$ is the size of the answer $q$ returns.
\end{restatable}
Compared to~\cite{CW89}, the second bound has an additional term that depends on the degrees of the points in $A_q$: if the center of a star $T$ lies in $A_q$, then $\sigma(q,T)=0$, while consecutive leaves in a traversal of $T$ may alternate between $P_q^{\mathrm{near}}$ and $P_q^{\mathrm{far}}$.
\begin{proof}
Root $T$ at an arbitrary point and let $\ppath=(v_1,\ldots,v_n)$ be the order in which a depth-first traversal of $T$ visits the points of $P$ for the first time. For each $i\in[n-1]$, let $\pi_i$ be the path of $T$ between $v_i$ and $v_{i+1}$. Since the parent of $v_{i+1}$ lies on the path from $v_i$ to the root, $\pi_i$ is the part of the closed walk performed by the traversal between its first visits to $v_i$ and $v_{i+1}$. These parts are edge-disjoint, and the closed walk traverses every edge of $T$ twice; hence, every edge of $T$ belongs to at most two paths $\pi_i$.  Moreover, the walk enters every interior point $v$ of $\pi_i$ by moving up from a child of $v$, since the only downward step of $\pi_i$ ends at $v_{i+1}$. The walk moves up from each child of $v$ to $v$ exactly once; hence, $v$ is an interior point of at most $\deg_T(v)$ paths $\pi_i$.

Now, fix $q$ and let $\{v_i,v_{i+1}\}$ be $\eps$-stabbed by $q$, where w.l.o.g.~$v_i\in P_q^{\mathrm{near}}$ and $v_{i+1}\in P_q^{\mathrm{far}}$. Let $u$ be the first point of $\pi_i$, starting from $v_i$, that does not belong to $P_q^{\mathrm{near}}$, and let $u'$ be the point preceding $u$ on $\pi_i$. Then, the edge $\{u',u\}$ of $T$ is weakly $\eps$-stabbed by $q$; charging $\{v_i,v_{i+1}\}$ to $\{u',u\}$ yields $\sigma(q,\ppath)\leq 2\tau(q,T)$. Furthermore, if $u\in P_q^{\mathrm{far}}$ then $\{u',u\}$ is $\eps$-stabbed by $q$; otherwise, $u\in A_q$, and $u$ is an interior point of $\pi_i$, since $v_{i+1}\notin A_q$. We charge $\{v_i,v_{i+1}\}$ to $\{u',u\}$ in the former case and to $u$ in the latter case, which yields $\sigma(q,\ppath)\leq 2\sigma(q,T)+\sum_{v\in A_q}\deg_T(v)$. The last inequality follows from $|A_q|\leq t_q$. The running time is $O(dn)$, since we only run a depth-first traversal of $T$.
\end{proof}

\begin{restatable}{lemma}{lempathtree}
\label{lem:pathtree}
    Given a spanning path $\ppath$ 
   on a pointset $P\subset \RR^d$ we can build a binary partition tree $\ptree$ on $P$, in $(dn\log n)$ time, such that for any point $q$,
   we have 
   $\zeta(q,\ptree)\leq 2[\sigma(q,\ppath)+t_q]\cdot \log n+1$, where $t_q = |\left(B(q,1+\eps)\setminus B(q,1)\right)\cap P|$.
\end{restatable}
\begin{proof}
We build a complete binary tree, and assign the points along the path to the leaves from left to right sequentially. This takes $O(dn\log n)$ time. 

Let $v$ be an inner node. 
For the two children nodes of $v$ to be visited by $q$, one of the following must be true: \begin{itemize}
    \item the subset of points $P_v$ assigned to the leaves of the subtree rooted at $v$ is $\eps$-stabbed by $q$
    \item there exists a $p\in P_v$ such that $p\in B(q,1+\eps)\setminus B(q,1)$.
\end{itemize}

If $P_v$ is $\eps$-stabbed by $q$ then, since the points of $P_v$ are consecutive along $\ppath$, either an edge of $\ppath$ between two points of $P_v$ is $\eps$-stabbed by $q$, or $P_v$ contains a point of $A_q\subseteq B(q,1+\eps)\setminus B(q,1)$ (consider the first point that is not in $P_q^{\mathrm{near}}$ when walking along $\ppath$ from a point of $P_v\cap P_q^{\mathrm{near}}$ towards a point of $P_v\cap P_q^{\mathrm{far}}$). Moreover, every edge of $\ppath$ and every point of $P$ is contained in the subsets of at most $\log n$ inner nodes.
As there are at most $\sigma(q,\ppath)$ edges that are $\varepsilon$-stabbed by $q$ and at most $t_q$ points that can be on the annulus of $q$, then there are at most $[\sigma(q,\ppath)+t_q]\cdot \log n$ inner nodes that are visited by $q$. When an inner node is $\eps$-visited, we continue our search to both children. These are all the nodes that can be visited, plus the root. Therefore, the number of nodes visited for a query point $q$ can be at most $2[\sigma(q,\ppath)+t_q]\cdot \log n+1$.
\end{proof}

We combine these results to get a bound on the visiting number of a partition tree. %The proof is straightforward and can be found in \Cref{app:ptreevis}.

\begin{restatable}{corollary}{corptreevis}\label{cor:ptreevis}
    Let $P\subset \RR^d$ be a set of $n$ points. 
    For any $c>0$, we can build a binary partition tree $\ptree$ on $P$, in $\poly(n, d,\frac{1}{\eps}
) + \tilde{O}(n)\cdot O(1/\eps)^d$ time such that with probability $1-1/n^c$ every $q\in \cG_{\gamma}^d$, with $\gamma=\eps/\sqrt{d}$, visits $\zeta(q,\ptree)=\tilde{O}(dn^{1-\rho}+t_q)$ nodes, where $t_q = |\left(B(q,1+\eps)\setminus B(q,1)\right)\cap P|$ and $\rho = \Theta(\eps^2/\log(1/\eps))$.
\end{restatable}
\begin{proof}
Using \Cref{thm:treeexists}, 
for any $c>0$, we can compute a spanning tree $T$ with maximum degree $\Delta(T)=O(\log n)$ in $ \poly(n, d,\frac{1}{\eps}
) + \tilde{O}(n)\cdot O(1/\eps)^d$ time, such that with probability of success $1-n^{-c}$, for any $q\in \cG_{\gamma}^d$, $\sigma(q,T)=\tilde{O}(dn^{1-\rho})$. We then transform $T$ into a spanning path $\ppath$ using \Cref{lem:treepath}, and finally transform it into a binary partition tree $\ptree$ using \Cref{lem:pathtree}:
    \begin{align*}
        \zeta(q,\ptree) &\leq 2[\sigma(q,\ppath)+t_q]\cdot \log n+1\leq 2[2\sigma(q,T)+(\Delta(T)+1)\cdot t_q]\cdot \log n +1\\ &\leq \tilde{O}(dn^{1-\rho})+O(t_q\log^2 n)=\tilde{O}(dn^{1-\rho}+t_q).
    \end{align*}
    %Which follows by applying \Cref{lem:pathtree}, \Cref{lem:treepath,thm:treeexists}.
\end{proof}

Finally, we present the main contribution of our paper, which states that for the approximate spherical range counting problem, we can build a data structure in polynomial time that supports sublinear queries (if $t_q$ is also sublinear) and takes near-linear space. In our proof, we make use of \Cref{thminternalnode,thm:treeexists,lem:pathtree}.

\begin{theorem}
\label{thm:final}
Let $P\subset \RR^d$ be a set of $n$ points, such that each $p\in P$ is assigned a weight $w_p\in \RR$, and let $\eps\in (0,1)$. We can build a data structure for the approximate range counting problem with input $P,\ \{w_p\mid p\in P\},\ \eps$, that uses space in $\tilde{{O}}(n)$, needs preprocessing time in $O(dn)+n^{O(\eps^{-2}\log(1/\eps))}$
and for each query $q\in \RR^d$, it returns an answer in 
    $ n^{1-\Theta(\eps^4/\log(1/\eps))}+t_q^{\varrho}\cdot n^{1-\varrho}$ time, where $\varrho=\Theta(\varepsilon^2)$ and $t_q = |\left(B(q,1+\eps)\setminus B(q,1)\right)\cap P |$.
    For each query, the preprocessing succeeds with probability $1-n^{-c}$, where $c>0$ is a constant hidden in the aforementioned time and space complexities. 
\end{theorem}
\begin{proof}
We first apply a JL transformation~\cite{JL84} to reduce the dimension to $O(\eps^{-2}\log n)$. By \Cref{lem:JL}, and by applying a union bound over $n$ vectors, for any $c>0$ and any query $q\in \RR^d$, with probability $1-1/n^{c}$, all distances from $P$ are approximately preserved within factors $1\pm \eps/10$. This initial step costs $O(dn)$ time. For the rest of the proof, we assume that all points lie in $\RR^{O(\eps^{-2}\log n)}$. Then, we assume that queries are snapped to the grid $\cG_{\gamma}^d$, $\gamma =\eps/(10\sqrt{d})$. This induces an additive $\pm \eps/10$ error to the distances between $P$ and the query. Hence, points within distance $1$ from the query may be moved to a distance up to $1+\eps/5$, and points at a distance of at least $1+\eps$ may be moved to a distance of at least $1+4\eps/5$. Notice that $1+\eps/2\leq \frac{1+4\eps/5}{1+\eps/5}$, hence it suffices to use a more refined error of $\eps/2$ for the rest of the construction and assume that the space is rescaled by a factor $1/(1+\eps/5)$.  
%\yiannis{not super happy here}

We build a binary partition tree $\ptree$ on the input with error parameter $\eps/2$ using \Cref{cor:ptreevis}. The construction of $\ptree$ can be done in $n^{O(\eps^{-2}\log(1/\eps))}$ time and be successful for a fixed $q$, in that $\zeta(q,\ptree)=\tilde{O}(n^{1-\Theta(\eps^2/\log(1/\eps))}+t_q)$, with  probability $1-1/n^c$ for any constant $c>0$. 
%as follows: we first compute a spanning tree using \Cref{thm:treeexists}, then transform it into a spanning path using \Cref{lem:treepath}, and finally transform it into a binary partition tree using \Cref{lem:pathtree}. 
For each internal node $v$, let $P_v$ be the set of points in the leaves of the subtree rooted at $v$. 
In each internal node $v$, we build a data structure $\cD_v$ as in \Cref{thminternalnode} on $P_v$ with error parameter $\eps/2$, and we store the cumulative weight $s_v$ of the points in $P_v$, i.e., $s_v = \sum_{p\in P_v} w_p $. 
%By \Cref{thm:treeexists} (where we set $d\gets O(\eps^{-2}\log n)$), and the fact that \Cref{lem:treepath} and \Cref{lem:pathtree} apply near-linear time transformations, 
By \Cref{thminternalnode}, each construction of a data structure $\cD_v$, can be done in $\tilde{O}(d|P_v|)$ time and be successful with probability $1-1/n^c$ for any constant $c>0$. 
By a union bound over all inner nodes,  we claim that all data structures $\cD_v$ are also successful for $q$ with probability $1-\Theta\left(\frac{n\log n}{n^c}\right)$. By \Cref{thminternalnode}, the time and space needed to build data structures for the inner nodes is near-linear per level of the tree (since each level defines a partition). 
%Space usage follows from \Cref{thminternalnode} and the fact that each level of the tree defines a partition. 

To answer a query $q\in \RR^d$, we traverse the tree and collect all cumulative weights $s_v$ in each internal node $v$, whose associated data structure $\cD_v$ returns ``covered". We continue traversing the children of a node $v$ if and only if  $\cD_v$ returns ``stabbed". By \Cref{cor:ptreevis}, %\Cref{thm:treeexists},  \Cref{lem:treepath} and \Cref{lem:pathtree}, 
assuming that the tree has been constructed correctly for  $q\in \RR^d$, the number of nodes visited during the query algorithm for input $q\in \RR^d$ is 
$\tilde{O}(n^{\alpha}+t_q)$, where $\alpha = 1- \Theta(\eps^2/\log(1/\eps))$. By \Cref{thminternalnode}, the time needed to query a data structure $\cD_v$ in an internal node $v$ is  
$\tilde{O}(|P_v|^{\beta})$, where $\beta = 1-\Theta(\eps^2)$.
At each level $i$ of $\ptree$, the number of nodes whose data structures are queried is upper bounded by both $2^i$ and $\zeta(q,\ptree)$.  
Hence, the running time at level $i$ for query $q$ is $\tilde{O}(T_i^q(n))$, where $T_i^q(n)$ is defined as follows:
\[
T_i^q(n) = \min\Bigl\{\, 2^{i},
\; n^{\alpha} +t_q   \Bigr\}\cdot \left(\frac{n}{2^{i}}\right)^{\beta} \leq (n^{\alpha}+t_q)^{1-\beta} \cdot n^{\beta},
\] 
because if $i\leq \log(n^{\alpha}+t_q)$ then 
$\min\Bigl\{\, 2^{i},
\; n^{\alpha} +t_q   \Bigr\} = 2^i$ and if $i>\log(n^{\alpha}+t_q)$ then 
$\min\Bigl\{\, 2^{i},
\; n^{\alpha} +t_q   \Bigr\} = n^{\alpha} +t_q$. 
Now, since $1-\beta <1$, 
\[
T_i^q(n)\leq n^{\alpha(1-\beta)+\beta}+t_q^{1-\beta}\cdot n^{\beta}
=    n^{1-\Theta(\eps^4/\log(1/\eps))}+t_q^{\varrho}\cdot n^{1-\varrho},
\]
where $\varrho=1-\beta$.
We observe that the query algorithm will not explore any descendants of an internal node $v$ 
if $\cD_v$ returns ``disjoint" when $\min_{x\in P_v} \|x-q\| >1$, or 
if $\cD_v$ returns ``covered" when $\max_{x\in P_v} \|x-q\| \leq 1+\eps/2$, in which case it takes into account $s_v$. Thus, the answer is calculated by summing over cumulative weights $s_{v_1},\ldots,s_{v_{\ell}}$ of nodes corresponding to subsets $P_{v_1},\ldots, P_{v_{\ell}}$ such that 
$B(q,1)\cap P\subseteq \bigcup_{i\in [\ell]}P_{v_i} \subseteq B(q,1+\eps/2)\cap P$. 
% \yiannis{under construction: Query time}
%     visiting number = $\tilde{O}(n^{\alpha}+t_q)$, $\alpha = 1- \Theta(\eps^2/\log(1/\eps))$, and 
% query time per node = $\tilde{O}(n^{\beta})$, where $\beta = 1-\Theta(\eps^2)$
% The running time at level $i$ for query $q$ is $\tilde{O}(T_i^q(n))$, where $T_i^q(n)$ is defined as follows:
% \[
% T_i^q(n) = \min\Bigl\{\, 2^{i},
% \; n^{\alpha} +t_q   \Bigr\}\cdot \left(\frac{n}{2^{i}}\right)^{\beta}.
% \]
% \[
% \text{If } i \ge \log (n^{\alpha} +t_q),
% \quad\text{then}\quad
% T_i \leq (n^{\alpha}+t_q)^{1-\beta} \cdot n^{\beta}
% \]
% \[
% \text{If } i < \log (n^{\alpha} +t_q),
% \quad\text{then}\quad
% T_i \leq (n^{\alpha}+t_q)^{1-\beta} \cdot n^{\beta}.
% \]
% Since $1-\beta <1$, 
% \begin{align*}
% T_i&\leq n^{\alpha(1-\beta)+\beta}+t_q^{1-\beta}\cdot n^{\beta}\\
% &=    n^{1-\Theta(\eps^4/\log(1/\eps))}+t_q^{\Theta(\eps^2)}\cdot n^{1-\Theta(\eps^2)}
% \end{align*}
\end{proof}

\section{Data-driven Range Searching}\label{sec:datadr}

%\andreas{TODO: check whether visiting number accounts for children of stabbed node sets AND children of nodesets thay have all points are in 1+eps or all more than 1}
As we saw in the previous section, we can construct the partition tree in polynomial time. However, the exponent depends on $1/\eps$, which we would like to avoid. In this section, we present an algorithm inspired by learning theory that constructs a partition tree with near-optimal visiting number in expectation, over the distribution of queries. Specifically, the algorithm will have access to a number $m$ of sample queries, and will construct a partition tree that is optimal for that sample - thus the algorithm is data-driven. Showing that uniform convergence is achieved, we prove that this partition tree is nearly optimal, in that its expected visiting number is nearly minimal.

In the next lemma, we show that using only a near-linear number of queries we can achieve uniform convergence for the expected $\eps$-stabbing number, i.e., for any spanning tree, the estimation of the $\eps$-stabbing number using only these queries is close to its real expected $\eps$-stabbing number. The same holds for the weak $\eps$-stabbing number, which we use to construct our partition tree.

\begin{restatable}{lemma}{lemmaspanningtreespreserved}
\label{lemma:spanningtreespreserved}
Let $P$ be a set of $n$ points in $\mathbb{R}^d$ and let $\cD_Q$ be a ball query distribution having as support a countably infinite set $\cR$.  
%Let $D$ be the dual VC dimension of the range space $(P,\cR)$.  
It suffices to sample $m \in O(n\cdot (d\log n+\log(1/\delta)))$ queries $S$ independently from $\cD_Q$, so that with probability at least $1-\delta$,   
for any spanning tree $\stree$ of $P$ and any $\phi\in\{\sigma,\tau\}$, $m^{-1}\cdot  \sum_{q\in S} \phi(q,\stree)$ is in
$
 \left[\frac{5}{8}\cdot \EE_{q\sim\cD_Q} \left[\phi(q,\stree)\right] - \frac{3}{8} , \frac{11}{8}\cdot \EE_{q\sim\cD_Q} \left[\phi(q,\stree)\right] + \frac{3}{8} \right]$.
\end{restatable}
\begin{proof}
    We apply \Cref{theo:etaepsilonapproxfunctionsrelative} with $\eta = 1/n$ and failure probability $\delta/2$ on each of the sets of functions $\cF$ and $\cF_w$ defined by the pairs of points $P\times P$ and the domain $\mathbb{R}^d$ corresponding to the set of ball queries $\cR$. For each $a,b\in P$, we have a function $f_{ab}\in \cF$ such that for any query centered at $q$: $f_{ab}(q) = 1 \iff $ $q$ $\eps$-stabs $\{a,b\}$, and a function $g_{ab}\in \cF_w$ such that $g_{ab}(q) = 1 \iff $ $q$ weakly $\eps$-stabs $\{a,b\}$. It remains to bound the VC dimensions of the range spaces $(\cR, \cF)$ and $(\cR, \cF_w)$. Let $\cF^- = \{f_a^-\mid  a\in P\}$ ($\cF^+ = \{f_a^+\mid  a\in P\}$) be a set of functions defined such that for any query $q$: $f_a^-(q)=1\iff \|a- q\|\leq 1$ ($f_a^+(q)=1\iff \|a- q\|< 1+\eps$).
    Notice that $f_{ab}(q)=1 \iff (f_a^-(q)=1 \land f_b^+(q)=0) \lor (f_a^+(q)=0 \land f_b^-(q)=1)$. Moreover, $a\in P_q^{\mathrm{near}}\iff f_a^-(q)=1$ and $a\in P_q^{\mathrm{far}}\iff f_a^+(q)=0$, hence $g_{ab}(q)$ is determined by $f_a^-(q)$, $f_a^+(q)$, $f_b^-(q)$, and $f_b^+(q)$.
    The dual VC dimension of $(P,\cR)$ is $d+1$ (\cite{AB02}), meaning that the VC dimension of $(\cR,\{f_a^- \mid a\in P\})$ and $(\cR,\{f_a^+ \mid a\in P\})$ is also $d+1$, implying that the VC dimensions of $(\cR, \cF)$ and $(\cR, \cF_w)$ are $O(d)$ (since each function of $\cF\cup\cF_w$ is defined as $O(1)$ and/or operations on functions of $\cF^-$ and $\cF^+$). Let $S=(x_1,\ldots,x_m)$ be as in \Cref{theo:etaepsilonapproxfunctionsrelative}, i.e., each $x_i$ is a query sampled independently at random from $\cD_Q$, and $m\in O\left(n \left( d \log \left(n\right) + \log\left(\frac{1}{\delta}\right) \right)\right)$. In the following, we prove the claim for $\phi=\sigma$, using the functions $f_{ab}$; the proof for $\phi=\tau$ is identical, using the functions $g_{ab}$ and weakly $\eps$-stabbed pairs instead. 

    By \Cref{theo:etaepsilonapproxfunctionsrelative}, with probability at least $1-\delta$ over the choice of $S$, for all ``light" pairs of points, i.e., pairs that are $\eps$-stabbed with probability at most $1/n$, we can estimate the probability through $S$ with an additive error of $3/(8n)$ whereas for all ``heavy'' pairs of points, i.e., those that are $\eps$-stabbed with probability more than $1/n$, we can estimate the probability of getting $\eps$-stabbed with a multiplicative error of $1\pm 3/8$. 
    Now let $\stree$ be any spanning tree of $P$ and let $E(\stree)$ be the set of its edges. Let $E_{\ell}\subseteq E(\stree)$ be the set of light edges of $\stree$, i.e., edges $(a,b)\in E(\stree)$ for which the probability that $\{a,b\}$ is $\eps$-stabbed is at most $1/n$ and let $E_{h}\subseteq E(\stree)$ be the set of heavy edges of $\stree$, i.e., edges $(a,b)\in E(\stree)$ for which the probability that $\{a,b\}$ is $\eps$-stabbed is more than $1/n$. Let 
    $\tilde{\sigma}(S,\stree)=
        m^{-1}\cdot {\sum_{i=1}^m \sigma(x_i,\stree)}$. 
    We have,
    \begin{align*}
        \tilde{\sigma}(S,\stree)=
        \frac{1}{m} \cdot \sum_{i=1}^m \sigma(x_i,\stree)
        &= \frac{1}{m} \cdot \sum_{i=1}^m \sum_{e\in E(\stree)}\sigma(x_i,e) \\&= 
   \frac{1}{m} \cdot \sum_{i=1}^m \left(\sum_{e\in E_{\ell}}\sigma(x_i,e)+\sum_{e\in E_{h}}\sigma(e,x_i) \right)\\
    & = \sum_{(a,b)\in E_{\ell}} \mu_{f_{ab}}(S) +  \sum_{(a,b)\in E_{h}} \mu_{f_{ab}}(S)\\
    %& \geq      \sum_{(a,b)\in E_{\ell}}     \left(\EE_{q\sim\cD_Q}[f_{ab}(q)] -\frac{3}{8n}\right)    +  \sum_{(a,b)\in E_{h}}     \frac{5\cdot \EE_{q\sim\cD_Q}[f_{ab}(q)]}{8}\\
    &\geq 
    \frac{5}{8}\cdot \EE_{q} \left[\sigma(q,\stree)\right] - \frac{3}{8} .
% \sum_{(a,b)\in E_{\ell}} 
%     \left(\EE_{q\sim\cD}[f_{ab}(q)] +\frac{1}{3n}\right)
%     +  \sum_{(a,b)\in E_{h}} 
%     \frac{11\cdot \EE_{q\sim\cD}[f_{ab}(q)]}{8}
%     \right]
    \end{align*}
   Similarly,
      \begin{align*}
       \tilde{\sigma}(S,\stree)= \frac{1}{m} \cdot \sum_{i=1}^m \sum_{e\in E(\stree)}\sigma(x_i,e) &\leq  
    \sum_{(a,b)\in E_{\ell}} 
    \left(\EE_{q}[f_{ab}(q)] +\frac{3}{8n}\right)+
       \sum_{(a,b)\in E_{h}} 
    \frac{11}{8}\cdot \EE_{q}[f_{ab}(q)]\\
    &\leq 
    \frac{11}{8}\cdot \EE_{q\sim\cD_Q} \left[\sigma(q,\stree)\right]  + \frac{3}{8} .
    \end{align*}
    Hence, for any spanning tree $\stree$ of $P$, $\tilde{\sigma}(S,\stree)$ is in
    \[\left[\frac{5}{8}\cdot \EE_{q\sim\cD_Q} \left[\sigma(q,\stree)\right] - \frac{3}{8} , \frac{11}{8}\cdot \EE_{q\sim\cD_Q} \left[\sigma(q,\stree)\right] + \frac{3}{8} \right].\] 
\end{proof}

The next lemma states that the number of nodes of a partition tree visited during a query search consists an upper bound on the $\eps$-stabbing number of a corresponding spanning path of the points for the same query. Its proof is a straightforward adaptation of the third item of \cite[Lemma 3.1]{CW89} and is moved to \Cref{app:ubstabbing}.

\begin{restatable}{lemma}{lemubstabbing}\label{lem:ubstabbing}
    If $\ptree$ is a partition tree for P, then there exists a spanning path $\ppath$ whose $\eps$-stabbing number does not exceed the visiting number of $\ptree$. 
\end{restatable}

We now prove the main result of this section.

\begin{theorem}
\label{thm:samplingcomplexity}
    Let $P$ be a set of $n$ points in $\RR^d$ and let $\cD_Q$ be a ball query distribution having as support a countably infinite set $\cR$. 
    It suffices to sample a set $S$ of $m \in O(n\cdot (d\log n+\log(1/\delta)))$ queries independently from $\cD_Q$, so that with probability at least $1-\delta$, we can compute a partition tree $\tilde{\ptree}$ such that 
\[
\EE_{q\sim \cD_Q}\left[ \zeta(q,\tilde{\ptree})\right] = O\left(\left[\min_{\ptree} \EE_{q\sim \cD_Q}[ \zeta(q,\ptree)]+\EE_{q\sim \cD_Q}[t_q]\right]\cdot \log n\right),
\]
where $t_q= |\left(B(q,1+\eps)\setminus B(q,1)\right)\cap P|$. Moreover, given $S$, the running time to compute $\tilde{\ptree}$ is in
$O(n^3 (d\log n+\log(1/\delta)))$. 
\end{theorem}
\begin{proof}
    By \Cref{lemma:spanningtreespreserved}, with probability at least $1-\delta$,
for any spanning tree $\stree$ of $P$, $m^{-1}\cdot  \sum_{q\in S} \tau(q,\stree)$ is in
\[
  \left[\frac{5}{8}\cdot \EE_{q\sim\cD_Q} \left[\tau(q,\stree)\right] - \frac{3}{8} , \frac{11}{8}\cdot \EE_{q\sim\cD_Q} \left[\tau(q,\stree)\right] + \frac{3}{8} \right].
\]
Now, let $E(\stree)$ denote the edges of spanning tree $\stree$ and $\tilde{\stree}=\mathop{\arg\min}_T \{m^{-1} \cdot \sum_{q\in S} \tau(q,\stree)\}$. Notice that
$\sum_{q\in S} \tau(q,\stree) = \sum_{q\in S} \sum_{e\in E(\stree)} \tau(q,\{e\})=  \sum_{e\in E(\stree)} \sum_{q\in S} \tau(q,\{e\})$ for any spanning tree $\stree$.
Hence, we can compute $\tilde{\stree}$, by first constructing a complete weighted graph on $P$ where we assign weight $\sum_{q\in S} \tau(q,\{e\})$ to each edge $e\in {{P} \choose {2}}$, and then computing the minimum spanning tree on that graph, using Kruskal's algorithm. The running time of this step is $O(n^2 m)$, which corresponds to the step of building the graph.

Let $\ptree$ be any partition tree of $P$, and let $\ppath$ be the spanning path given by \Cref{lem:ubstabbing} for $\ptree$, i.e., $\sigma(q,\ppath)\leq \zeta(q,\ptree)$ for any $q$. An edge of $\ppath$ that is weakly $\eps$-stabbed but not $\eps$-stabbed by $q$ has an endpoint in $A_q$, and every point of $P$ is an endpoint of at most two edges of $\ppath$; hence, $\tau(q,\ppath)\leq \sigma(q,\ppath)+2|A_q|\leq \zeta(q,\ptree)+2t_q$. Since $\ppath$ is also a spanning tree of $P$, by \Cref{lemma:spanningtreespreserved} we get:
\[
\frac{5}{8}\cdot\EE_{q\sim\cD_Q} [\tau(q,\tilde{\stree})]-\frac{3}{8}\leq m^{-1} \sum_{q\in S} \tau(q,\tilde{\stree})\leq m^{-1} \sum_{q\in S} \tau(q,\ppath)\leq \frac{11}{8}\cdot \EE_{q\sim\cD_Q} [\tau(q,\ppath)]+\frac{3}{8}.
\]
Since $\zeta(q,\ptree)\geq 1$ for any $q$, we obtain $\EE_{q\sim\cD_Q} [\tau(q,\tilde{\stree})] = O\left(\EE_{q\sim \cD_Q}[\zeta(q,\ptree)]+\EE_{q\sim \cD_Q}[t_q]\right)$.
Given $\tilde{\stree}$, we use \Cref{lem:treepath} to construct a spanning path $\tilde{\ppath}$ such that $\sigma(q,\tilde{\ppath})\leq 2\tau(q,\tilde{\stree})$ for any $q$, and then \Cref{lem:pathtree} to construct a binary partition tree $\tilde{\ptree}$ of $P$ such that
\begin{align*}
    \EE_{q\sim \cD_Q}[\zeta(q,\tilde{\ptree})] &\leq 2\left(\EE_{q\sim \cD_Q}[\sigma(q,\tilde{\ppath})] +\EE_{q\sim \cD_Q}[t_q]\right)\cdot\log n+1\\
    &\leq 2\left(2\,\EE_{q\sim \cD_Q}[\tau(q,\tilde{\stree})] +\EE_{q\sim \cD_Q}[t_q]\right)\cdot\log n+1\\
&= O\left(\left[\EE_{q\sim \cD_Q}[\zeta(q,\ptree)] +\EE_{q\sim \cD_Q}[t_q]\right]\cdot\log n\right).
\end{align*}
Since $\ptree$ is an arbitrary partition tree of $P$, the claimed bound follows. Both transformations take $O(dn\log n)$ time. Thus, the overall running time is $O(n^2 m ) \subseteq O(n^3 (d\log n+\log(1/\delta)))$.
\end{proof}

\bibliography{jlann}

\appendix
\section{Supplementary Material for \Cref{sec:prel}}

\subsection{Proof of \Cref{theo:etaepsilonapproxfunctionsrelative}}
\label{app1}
\corapproxfunctions*
\begin{proof}
 We first consider the case $ \EE_{q\sim \cD} [f(q)] \leq \eta$. 
By \Cref{theo:etaepsilonapproxfunctions}, with $ \eps = 1/9$, we obtain 
\[d_{\eta}\left(\EE_{q\sim \cD}[f(q)],\mu_f(S)\right)\leq  1/9.\] Hence, 
    $\left|\EE_{q\sim \cD}[f(q)]-\mu_f(S)\right|\leq  (1/9) \cdot (\eta + \EE_{q\sim \cD}[f(q)] + \mu_f(S))\leq \frac{2\eta}{9} + \frac{\mu_f(S)}{9}$ which implies $\left|\EE_{q\sim \cD}[f(q)]-\mu_f(S)\right|\leq \frac{\eta}{3}$, since 
    \[
    \EE_{q\sim \cD}[f(q)]\geq \mu_f(S)
    \implies \left|\EE_{q\sim \cD}[f(q)]-\mu_f(S)\right|\leq \frac{2\eta}{9} + \frac{\EE_{q\sim \cD}[f(q)]}{9} \leq \frac{\eta}{3},
    \]
    and 
    \begin{align*}
        \EE_{q\sim \cD}[f(q)]< \mu_f(S)
    &\implies
       \left|\EE_{q\sim \cD}[f(q)]-\mu_f(S)\right|= \mu_f(S) - \EE_{q\sim \cD}[f(q)] \leq \frac{2\eta}{9} + \frac{\mu_f(S)}{9}\\
    &\implies
        \left|\EE_{q\sim \cD}[f(q)]-\mu_f(S)\right|=\mu_f(S) - \EE_{q\sim \cD}[f(q)] \leq \frac{3\eta}{8}.
    \end{align*}
    Next, we consider the case $\EE_{q\sim \cD} [f(q)] > \eta$. By \Cref{theo:etaepsilonapproxfunctions} with $\eps =1/9$,  $\left|\EE_{q\sim \cD}[f(q)]-\mu_f(S)\right|\leq  \frac{1}{9} \cdot \left(\eta + \EE_{q\sim \cD}[f(q)] + \mu_f(S)\right)\leq \frac{ 2\cdot\EE_{q\sim \cD}[f(q)]}{9} + \frac{\mu_f(S)}{9}$, which implies the claim since:
    \[
    \EE_{q\sim \cD}[f(q)]\geq \mu_f(S)
    \implies 
    \left|\EE_{q\sim \cD}[f(q)]-\mu_f(S)\right|\leq  \frac{1}{3}\cdot \EE_{q\sim \cD}[f(q)],
    \]
    and
    \begin{align*}
    \EE_{q\sim \cD}[f(q)]< \mu_f(S)
    &\implies \left|\EE_{q\sim \cD}[f(q)]-\mu_f(S)\right|= \mu_f(S) - \EE_{q\sim \cD}[f(q)] \leq \frac{3}{8}\cdot \EE_{q\sim \cD}[f(q)].   
    \end{align*}
\end{proof}

\subsection{Proof of \Cref{lemma:multisampling}}
\label{app2}
\lemmamultisampling*
\begin{proof}
 We preprocess a balanced binary tree on values $\gamma_1,\ldots,\gamma_{n}$ where each internal node $v$ also stores the sum $s_v$ of all leaves of the subtree rooted at $v$. The time needed to compute the tree is $O(n\log n)$, which also bounds its space usage. 

    To sample a point from $P$, we sample uniformly at random an integer $z\in [\sum_{i} \gamma_i]$ and we search the binary tree for the left-most leaf with value, say, $\gamma_j$ such that $\sum_{i\leq j} \gamma_i \geq z$. We return $p_j$. The time needed to traverse the tree is $O(\log n)$.

    To update the weight of a point $p_i\in P$, it suffices to update the sums stored in the ancestors of the leaf storing $\gamma_i$. Assuming that $\gamma_i$ is being replaced by $\gamma_i'$, we traverse the $O(\log n)$ ancestors, and we add $\gamma_i' -\gamma_i$ in each one of them. This requires $O(\log n)$ time in total.
\end{proof}

\clearpage
\section{Supplementary Material for \Cref{section:stabbing}}

\subsection{Pseudocode of \Cref{thm:stabds}}
\label{app:pseudocode}

\begin{algorithm}[H]
	\caption{Query}
	\label{alg:query}
	\begin{algorithmic}[1]
    \Statex \textbf{Input:} $q \in \RR^d, d' = \frac{\log |P|}{1+\eps^2}$ 
    \State $L \gets 100 n^{1-\beta}$, $\beta = \frac{\eps^2}{115200(1+\eps^2)}$
    \State $\mathcal{I}_{out}\gets \emptyset$
    \State $m\gets 0$
    \State $a \gets False$
    \For{$x \in \{0,1\}^{d'}:~\|x-f(q)\|_1\leq \theta$ } \Comment{search for a near point}
        \For{$i \in \mathcal{D}[x]$} \Comment{iterate over indices in the list}
            \State $m\gets m+1$
            \If{$\|p_i-q\|\leq 1+\eps$}
            \State $a\gets True$
            \State \textbf{break}
            \ElsIf{$m>L$} \Comment{too many false positives: stop this search}
            \State \textbf{break}
            \EndIf
            \EndFor
        \If{$a$ \textbf{or} $m>L$} \Comment{near point found, or limit exceeded}
            \State \textbf{break}
        \EndIf
    \EndFor
    \If{a}
        \State $\mathcal{I}_{out}\gets \mathcal{I}_{out}\cup \{i\}$
    \EndIf
    \State $m\gets 0$ \Comment{the far search has its own budget}
    \State $b\gets False$
     \For{$x \in \{0,1\}^{d'}:~\|x-f(q)\|_1\geq \theta$ } \Comment{search for a far point}
        \For{$i \in \mathcal{D}[x]$} \Comment{iterate over indices in the list}
            \State $m\gets m+1$
            \If{$\|p_i-q\|\geq 1$}
            \State $b\gets True$
            \State \textbf{break}
            \ElsIf{$m>L$} \Comment{too many false positives: stop this search}
            \State \textbf{break}
            \EndIf
            \EndFor
        \If{$b$ \textbf{or} $m>L$} \Comment{far point found, or limit exceeded}
            \State \textbf{break}
        \EndIf
    \EndFor
    \If{b}
        \State $\mathcal{I}_{out}\gets \mathcal{I}_{out}\cup \{i\}$
    \EndIf
    \State \Return $\mathcal{I}_{out}$        
    \end{algorithmic}
\end{algorithm}    

%\thmstabds*
\begin{algorithm}[H]
	\caption{Preprocessing}
	\label{alg:preprocessing}
	\begin{algorithmic}[1]
    \Statex \textbf{Input:} $P=\{p_1,\ldots, p_n\} \subset \RR^d$, $\eps \in (0,1)$ 
    \State Initialize dictionary $\mathcal{D}$
		\State Sample $f:\RR^d \to \{0,1\}^{d'}$ as in \Cref{lemma:embedtohamming}, $d' = \frac{\log |P|}{1+\eps^2}$
        \For{$p_i \in P$}
                \If{$\mathcal{D}[f(p_i)]$ is empty}
                    \State Initialize an empty list in $\mathcal{D}[f(p_i)]$
                    \State add $i$ to the list
                 \Else 
                 \State add $i$ to the list in $\mathcal{D}[f(p_i)]$
                \EndIf
          \EndFor      
          \State \Return $\mathcal{D}$
	\end{algorithmic}
\end{algorithm}

\section{Supplementary Material for \Cref{sec:datadr}}

\subsection{Proof of 
\Cref{lem:ubstabbing}}
\label{app:ubstabbing}

\lemubstabbing*

\begin{proof}
     We assign an arbitrary left-to-right order among the children of every internal node of $\ptree$, and let $l_1,\dots,l_n$ denote the leaves of $\ptree$ from left to right. Next, form the spanning path $(x_1,\dots, x_n)$ of $P$, where $\{x_i\}$ is the node-set of $l_i$. If the edge $\{x_i, x_{i+1}\}$ is $\eps$-stabbed by a query $q$, we charge this event to the unique child of the nearest common ancestor $v$ of $l_i$, and $l_{i+1}$ that is also an ancestor of $l_i$ (or $l_i$ itself). Since the pointset corresponding to $v$ is necessarily $\eps$-stabbed by $q$, the child which takes the charge is visited. Furthermore, such a node cannot be charged twice. 
\end{proof}

\end{document}